\documentclass[11pt]{article}
\usepackage[utf8]{inputenc}
\usepackage[english]{babel}
\usepackage{graphicx}         
\usepackage[pdftex, colorlinks=true, linkcolor=black, citecolor=black]{hyperref}
\usepackage{braket}
\usepackage{amsmath}
\usepackage{amssymb}
\usepackage{amsthm}
\usepackage{mathtools}
\usepackage{bbm}		
\usepackage{array}		
\usepackage{color}
\usepackage[usenames,dvipsnames]{xcolor}
\usepackage{caption}
\usepackage{enumitem}	

\allowdisplaybreaks

\theoremstyle{definition}\newtheorem{definition}{Definition}[section]

\theoremstyle{plain}\newtheorem{thm}{Theorem}
\newtheorem{lem}{Lemma}[section]
\newtheorem{cor}[lem]{Corollary}
\newtheorem{prop}[lem]{Proposition}

\theoremstyle{remark}\newtheorem{remark}{Remark}

\newcommand{\lemit}[1]{\begin{enumerate}[label={(\alph*)}, ref={\thelem\alph*}]{#1}\end{enumerate}}	%items in allen lem's können gelabelt werden
\newcommand{\remit}[1]{\begin{enumerate}[label={(\alph*)}, ref={\theremark\alph*}]{#1}\end{enumerate}}

\renewcommand{\hat}[1]{\widehat{#1}}											%hat
\renewcommand{\tilde}[1]{\widetilde{#1}}										%tilde

\newcommand{\ls}{\lesssim}													%smaller up to constant
\newcommand{\gs}{\gtrsim}													%greater up to constant

\newcommand{\Vp}{V^\parallel}
\newcommand{\phe}{\varphi^\varepsilon}
\newcommand{\chie}{\chi^\varepsilon}

\newcommand{\z}{z}															%z

\newcommand{\wb}{w_\beta}
\newcommand{\wbot}{w_\beta^{(12)}}
\newcommand{\bb}{b_\beta}													%b^beta

\newcommand{\he}{h_\varepsilon}												%h^epsilon
\newcommand{\heot}{\he^{(12)}}
\newcommand{\heoth}{\he^{(13)}}
\newcommand{\heij}{\he^{(ij)}}

\newcommand{\hbo}{\overline{h}_{\bo}}									%h^\beta_1
\newcommand{\hboot}{\overline{h}_\bo^{(12)}}
\newcommand{\hbooth}{\overline{h}_\bo^{(13)}}
\newcommand{\hboij}{\overline{h}_\bo^{(ij)}}

\newcommand{\owot}{\overline{w}^{(12)}}
\newcommand{\owij}{\overline{w}^{(ij)}}

\newcommand{\te}{\Theta_\varepsilon}											%Theta^\epsilon
\newcommand{\teij}{\te^{(ij)}}
\newcommand{\teot}{\te^{(12)}}
\newcommand{\teoth}{\te^{(13)}}

\newcommand{\tb}{\overline{\Theta}_{\bo}}									%Theta^\beta_1
\newcommand{\tbot}{\overline{\Theta}_\bo^{(12)}}
\newcommand{\tboth}{\overline{\Theta}_\bo^{(13)}}
\newcommand{\tbij}{\overline{\Theta}_\bo^{(ij)}}

\newcommand{\tz}{\overline{\Theta}_0}
\newcommand{\tzot}{\tz^{(12)}}

\newcommand{\hz}{\overline{h}_0}
\newcommand{\hzot}{\hz^{(12)}}

\newcommand{\hb}{\overline{h}_\beta}
\newcommand{\hbot}{\hb^{(12)}}

\newcommand{\tbeta}{\overline{\Theta}_\beta}
\newcommand{\tbetaot}{\tbeta^{(12)}}

\newcommand{\na}{\nabla}														%\nabla

\newcommand{\pp}{p^\Phi}														%p^Phi
\newcommand{\pc}{p^{\chie}}														%p^chi	
\newcommand{\qp}{q^\Phi}														%q^Phi
\newcommand{\qc}{q^{\chie}}														%q^chi

\newcommand{\bo}{{\beta_1}}													%\beta_1

\newcommand{\lr}[1]{\left\langle #1 \right\rangle} 							%Scalar Product
\newcommand{\llr}[1]{\left\llangle #1 \right\rrangle}							%double scalar product
\newcommand{\norm}[1]{\lVert#1\rVert}   										%norm
\newcommand{\onorm}[1]{\lVert#1\rVert_\mathrm{op}}							%operator norm

\newcommand{\R}{\mathbb{R}}													%R

\newcommand\mydots{,\makebox[1em][c]{.\hfil.\hfil.},}							%more compact \dots
\newcommand{\Tr}{\mathrm{Tr}}
\renewcommand{\d}{\mathop{}\!\mathrm{d}}

\DeclareMathOperator*{\supp}{\mathrm{supp}}

\makeatletter
\DeclareFontFamily{OMX}{MnSymbolE}{}
\DeclareSymbolFont{MnLargeSymbols}{OMX}{MnSymbolE}{m}{n}
\SetSymbolFont{MnLargeSymbols}{bold}{OMX}{MnSymbolE}{b}{n}
\DeclareFontShape{OMX}{MnSymbolE}{m}{n}{
    <-6>  MnSymbolE5
   <6-7>  MnSymbolE6
   <7-8>  MnSymbolE7
   <8-9>  MnSymbolE8
   <9-10> MnSymbolE9
  <10-12> MnSymbolE10
  <12->   MnSymbolE12
}{}
\DeclareFontShape{OMX}{MnSymbolE}{b}{n}{
    <-6>  MnSymbolE-Bold5
   <6-7>  MnSymbolE-Bold6
   <7-8>  MnSymbolE-Bold7
   <8-9>  MnSymbolE-Bold8
   <9-10> MnSymbolE-Bold9
  <10-12> MnSymbolE-Bold10
  <12->   MnSymbolE-Bold12
}{}

\let\llangle\@undefined
\let\rrangle\@undefined
\DeclareMathDelimiter{\llangle}{\mathopen}%
                     {MnLargeSymbols}{'164}{MnLargeSymbols}{'164}
\DeclareMathDelimiter{\rrangle}{\mathclose}%
                     {MnLargeSymbols}{'171}{MnLargeSymbols}{'171}
\makeatother

\title{Derivation of the 1d NLS equation from the 3d quantum many-body dynamics of strongly confined bosons}
\author{Lea Boßmann\thanks{Fachbereich Mathematik, Eberhard Karls Universität Tübingen\newline
	\indent\hspace{4pt} Auf der Morgenstelle 10, 72076 Tübingen, Germany\newline
	\indent\hspace{4pt} E-mail: lea.bossmann@uni-tuebingen.de}}
\date{\today}

\begin{document}
\maketitle

\begin{abstract}\noindent
We consider the dynamics of $N$ interacting bosons initially exhibiting Bose--Einstein condensation. Due to an external trapping potential, the bosons are strongly confined in two spatial directions, with the transverse extension of the trap being of order $\varepsilon$. 
The non-negative interaction potential is scaled such that its scattering length is positive and of order $(N/\varepsilon^2)^{-1}$, the range of the interaction scales as $(N/\varepsilon^2)^{-\beta}$ for $\beta\in(0,1)$.
We prove that in the simultaneous limit $N\rightarrow\infty$ and $\varepsilon\rightarrow 0$, the condensation is preserved by the dynamics and the time evolution is asymptotically described by a cubic defocusing nonlinear Schrödinger equation in one dimension, where the strength of the nonlinearity depends on the interaction and on the confining potential. This is the first derivation of a lower-dimensional effective evolution equation for singular potentials scaling with $\beta\geq\frac12$ and lays the foundations for the derivation of the physically relevant one-dimensional Gross--Pitaevskii equation ($\beta=1$) in \cite{GP}. For our analysis, we adapt an approach by Pickl \cite{pickl2008} to the problem with strong confinement.
\end{abstract}

\section{Introduction}
We consider a system of $N$ identical bosons in $\R^3$ interacting among each other through repulsive pair interactions. The bosons are trapped within a cigar-shaped trap, which effectively confines the particles in two spatial directions to a region of order $\varepsilon$. To describe this mathematically, let us first introduce the coordinates 
$$\z=(x,y)\in\R^{1+2}.$$
The cigar-shaped confinement is given by the scaled potential 
$\tfrac{1}{\varepsilon^2}V^\perp\left(\tfrac{y}{\varepsilon}\right)$
for some $0<\varepsilon\ll1$ and $V^\perp:\R^2\rightarrow\R$.
The Hamiltonian of this system is
\begin{equation}\label{H:micro:coord}
H_\beta(t)=\sum\limits_{j=1}^N\left(-\Delta_j+\tfrac{1}{\varepsilon^2}V^\perp\left(\tfrac{y_j}{\varepsilon}\right)+\Vp(t,\z_j)\right)+\sum\limits_{1\leq i<j\leq N}\wb(\z_i-\z_j),
\end{equation}
where $\Delta$ denotes the Laplace operator on $\R^3$ and $\Vp$ is a possibly time-dependent additional external potential. The units are chosen such that $\hbar=1$ and $m=\frac12$. 
%The lengths associated with the ground states of $-\partial_{x}^2$ and $-\Delta_y+V^\perp(y)$ are both of order one, hence $\varepsilon$ measures the transverse extension of the trap. 
In the limit $\varepsilon\rightarrow0$, the system becomes effectively one-dimensional, in the sense that excitations in the transverse direction are energetically strongly suppressed.

The interaction between the particles is described by the potential $\wb$ with scaling parameter $\beta\in(0,1)$. For the sake of this introduction, let us for the moment assume that 
$$\wb(\z)=\left(\tfrac{N}{\varepsilon^2}\right)^{-1+3\beta}w\left(\left(\tfrac{N}{\varepsilon^2}\right)^\beta\z\right)$$
for some compactly supported, spherically symmetric, non-negative, bounded potential $w$.\footnote{In our main result, the interaction is of a more generic form.} This scaling describes a dilute gas, where the scaling parameter $\beta$ interpolates between Hartree ($\beta=0$)  and Gross--Pitaevskii ($\beta=1$) regime. The proof of the physically relevant Gross--Pitaevskii regime relies essentially on the result for $\beta\in(0,1)$ and is given in \cite{GP}.
%We refer to \cite{GP} for a physical motivation of the scaling.
An important parameter characterising the interaction $\wb$ is its effective range,
$$\mu:=\left(\tfrac{N}{\varepsilon^2}\right)^{-\beta}.$$ 

We study the dynamics of the system in the simultaneous limit $(N,\varepsilon)\rightarrow(\infty,0)$.
The state $\psi^{N,\varepsilon}(t)$ of the system at time $t$ is determined by the $N$-body Schrödinger equation
\begin{equation}\label{SE}
i\tfrac{\d}{\d t}\psi^{N,\varepsilon}(t)=H_\beta(t)\psi^{N,\varepsilon}(t)
\end{equation}
with initial data $\psi^{N,\varepsilon}(0)=\psi^{N,\varepsilon}_0\in L^2_+(\R^{3N}):=\otimes_\mathrm{sym}^N L^2(\R^3).$
We assume that the system initially exhibits Bose--Einstein condensation, i.e.~that the one-particle reduced density matrix $\gamma^{(1)}_{\psi_0^{N,\varepsilon}}$ of $\psi^{N,\varepsilon}_0$,
\begin{equation}
\label{eqn:k:particle:RDM}
\gamma^{(k)}_{\psi^{N,\varepsilon}_0}:=\Tr_{k+1\mydots N}\ket{\psi^{N,\varepsilon}_0}\bra{\psi^{N,\varepsilon}_0}
\end{equation}
for $k=1$, is asymptotically close to the projection onto a one-body state $\phe_0$. At low energies, the state factorises as a consequence of the strong confinement and is of the form $\phe_0(z)=\Phi_0(x)\chie(y)\in L^2(\R^3)$ (see Remark \ref{rem:LSSY}).
Here, $\Phi_0$ denotes the wavefunction along the $x$-axis and $\chie$ is the normalised ground state of $-\Delta_y+\frac{1}{\varepsilon^2}V^\perp(\frac{y}{\varepsilon})$ in the confined directions. Due to the rescaling by $\varepsilon$, $\chie$ is given by
\begin{equation}\label{chie}
\chie(y)=\tfrac{1}{\varepsilon}\chi(\tfrac{y}{\varepsilon}),
\end{equation}
where $\chi$ is the normalised ground state of $-\Delta_y+V^\perp(y)$. 

In Theorem \ref{thm}, we show that if the state of the system is initially such a factorised Bose--Einstein condensate with condensate wavefunction $\phe_0=\Phi_0\chie$, i.e.~if
$$\lim\limits_{(N,\varepsilon)\to(\infty,0)}\Tr_{L^2(\R^3)}\left|\gamma^{(1)}_{\psi^{N,\varepsilon}_0}-\ket{\phe_0}\bra{\phe_0}\right|=0,
$$
where the limit $(N,\varepsilon)\to(\infty,0)$ is taken in an appropriate way,
then the condensation of the system into a factorised state is preserved by the dynamics, i.e. for all $t\in\R$ and $k\in\mathbb{N}$,
$$
\lim\limits_{(N,\varepsilon)\rightarrow(\infty,0)}\Tr_{L^2(\R^{3k})}\Big|\gamma^{(k)}_{\psi^{N,\varepsilon}(t)}-\ket{\phe(t)}\bra{\phe(t)}^{\otimes k}\Big|=0.
$$
%all $k$-particle reduced density matrices of $\psi^{N,\varepsilon}(t)$ remain asymptotically close to $\ket{\phe(t)}\bra{\phe(t)}^{\otimes k}$.
The condensate wavefunction at time $t$ is given by $\phe(t)=\Phi(t)\chie$, where $\Phi(t)$ is the solution of the one-dimensional nonlinear Schrödinger (NLS) equation
\begin{equation}\label{NLS}
i\tfrac{\partial}{\partial t}\Phi(t,x)=\left(-\tfrac{\partial^2}{\partial x^2}+\Vp(t,(x,0))+\bb|\Phi(t,x)|^2\right)\Phi(t,x)
=:h(t)\Phi(t,x)
\end{equation}
with $\Phi(0)=\Phi_0 $ and coupling parameter $\bb=\norm{w}_{L^1(\R^3)}\int_{\R^2}|\chi(y)|^4\d y.$ 

To our knowledge, Theorem \ref{thm} is the first rigorous derivation of an effectively lower-dimensional evolution equation directly from the three-dimensional $N$-body dynamics for $\beta\geq\frac12$. 
In \cite{keler2016}, von Keler and Teufel consider a similar problem for $\beta\in(0,\frac13)$ and in \cite{chen2013} and \cite{chen2017}, Chen and Holmer study interactions for values of $\beta$ in subsets of the interval $(0,\frac12)$.
The extension to $\beta\in(0,1)$ requires a non-trivial adaptation of methods used for the fully three-dimensional problem without strong confinement \cite{pickl2015} to handle the additional limit $\varepsilon\to0$ and the associated dimensional reduction. 
Not only is this an interesting mathematical problem on its own but it lays the foundations for the derivation of the physically relevant effectively one-dimensional Gross--Pitaevskii equation corresponding to the scaling $\beta=1$ \cite{GP}. In fact, the main idea of the proof in \cite{GP} is to approximate the interaction $w_{\beta=1}$ by a softer scaling interaction which is covered by our Theorem \ref{thm}, and to show that the remainders from this substitution vanish in the limit. The dimensional reduction occurs in the approximated interaction, hence the result for $\beta=1$ relies essentially on the tools and results proven here.

Let us give a brief motivation of the effective equation \eqref{NLS}. The $N$-body problem is interacting, hence the effective evolution is nonlinear and the strength of the linearity depends on the two-body scattering process. This process is to leading order described by the scattering length $a_\beta$ of $\wb$, which scales as $(\tfrac{N}{\varepsilon^2})^{-1}$ for $\beta\in(0,1]$ \cite[Lemma A.1]{erdos2007}. This implies that, for $\beta\in(0,1)$, the length scale of the inter-particle correlations is small compared to the range $\mu=(\tfrac{N}{\varepsilon^2})^{-\beta}$of $\wb$.
Hence, the correlations are negligible in the limit and the two-body scattering process is described by the first order Born approximation to the scattering length, $8\pi a_\beta\approx \int w_\beta(z)\d z$. 
The additional factor $\int_{\R^2}|\chi(y)|^4\d y$ in the coupling parameter arises from integrating out the transverse degrees of freedom in the course of the dimensional reduction. 
\\

Quasi one-dimensional Bose gases in highly elongated traps have been studied experimentally \cite{gorlitz2001,henderson2009} and the dynamical behaviour of such systems is of great physical interest \cite{esteve2006,kinoshita2006,meinert2017}. The first rigorous derivation of an NLS evolution for three-dimensional bosons was by Erd{\H o}s, Schlein and Yau \cite{erdos2007}. The main tool of their proof is the convergence of the BBGKY hierarchy, a system of coupled equations determining the time evolution of all $k$-particle density matrices. Later, the authors adapted their proof to handle the Gross--Pitaevskii scaling of the interaction \cite{erdos2010}. 
A different approach providing rates for the convergence of the reduced density matrices was proposed by Pickl \cite{pickl2008,pickl2011}, who derived effective evolution equations for NLS and Gross--Pitaevskii scaling of the interaction, including time-dependent external potentials \cite{pickl2015} as well as non-positive \cite{pickl2010,jeblick2018} and singular interactions \cite{knowles2010}. A third method for the Gross--Pitaevskii case, based on Bogoliubov transformations and coherent states on Fock space, was developed by Benedikter, De Oliveira and Schlein \cite{benedikter2015}, and a presumably optimal rate of convergence was recently proven by Brennecke and Schlein \cite{brennecke2017}. Further results concern bosons in one \cite{adami2007, chen2016} and two \cite{kirkpatrick2011,jeblick2016,jeblick2017} spatial dimensions.

Some authors have considered the problem of dimensional reduction for the NLS equation. In \cite{mehats2017}, M{\'e}hats and Raymond
study the cubic NLS equation in a two-dimensional quantum waveguide, i.e.~within a tube of width $\varepsilon$ around a curve in $\R^2$. They show that in the limit $\varepsilon\rightarrow0$, the nonlinear evolution is well approximated by a one-dimensional cubic NLS equation with an additional potential term due to the curvature. 
Ben Abdallah, Méhats, Schmeiser and Weishäupl consider in \cite{abdallah2005_2} an $(n + d)$-dimensional NLS equation subject to a strong confinement in $d$ directions and derive an effective $n$-dimensional NLS equation with a modified nonlinearity.

As mentioned above, there are few results concerning the derivation of lower-dimensional NLS equations from the underlying three-dimensional $N$-body dynamics. Chen and Holmer consider three-dimensional bosons with pair interactions in a harmonic potential that is strongly confining in one \cite{chen2013} or two \cite{chen2017} directions. For a repulsive interaction scaling with $\beta\in (0,\frac25)$ in case of the disc-shaped and for an attractive interaction with $\beta\in(0,\frac37)$ in case of the cigar-shaped confinement, they prove that the dynamics are effectively described by a two- or respectively one-dimensional NLS equation. 
In \cite{keler2016}, von Keler and Teufel study a Bose gas confined to a quantum waveguide with non-trivial geometry for scaling parameters $\beta\in(0,\tfrac13)$. They prove that the evolution is well captured by a one-dimensional NLS equation with additional potential terms arising from the twisting and bending of the waveguide. 
%The main difference to our work is that they treat repulsive pair interactions with scaling parameter $\beta\in(0,\frac13)$ whereas we extend their analysis to the regime $\beta\in(0,1)$.

The remainder of this paper is structured as follows: in Section \ref{sec:main}, we specify our assumptions and present the result. Our proof follows an approach by Pickl, which is outlined in Section \ref{sec:proof}. This section also contains the proof of our main Theorem \ref{thm}, relying essentially on two propositions. Finally, these propositions are proven in Section \ref{sec:prop}.

\section{Main Result}\label{sec:main}
To study the effective behaviour of the many-body system in the simultaneous limit $(N,\varepsilon)\rightarrow(\infty,0)$, let us consider families of initial data $\psi^{N,\varepsilon}_0$ along sequences $(N_n,\varepsilon_n)\rightarrow(\infty,0)$.
\begin{definition}
A sequence $(N_n,\varepsilon_n)$ in $\mathbb{N}\times(0,1)$ is called \textit{admissible} if
$$\lim\limits_{n\rightarrow\infty}(N_n,\varepsilon_n)=(\infty,0) \quad\text{ and }\quad \lim\limits_{n\rightarrow\infty}\frac{\varepsilon^2_n}{\mu_n}=0\quad\text{ for }\mu_n:=\left(\frac{N_n}{\varepsilon^2_n}\right)^{-\beta}.$$
It is called \textit{moderately confining} if 
$$\lim\limits_{n\rightarrow\infty}\frac{\mu_n}{\varepsilon_n}=0.$$ 
\end{definition}

Moderate confinement means that the extension $\varepsilon$ of the confining potential shrinks to zero but is still large compared to the range of the interaction $\mu$. This prevents the interaction from being supported mainly in a region that is quasi inaccessible to the particles due to the strong confinement. As $\mu/\varepsilon=N^{-\beta}\varepsilon^{2\beta-1}$, this condition is a restriction only for $\beta<\frac12$.

The admissibility condition ensures that $\varepsilon$ shrinks sufficiently fast compared to $\mu$ that the system becomes effectively one-dimensional. Note that for $\delta>0$, $\varepsilon^\delta/\mu=N^\beta\varepsilon^{\delta-2\beta}$, hence $\delta=2$ is the smallest exponent for which $\varepsilon^\delta/\mu\rightarrow0$ is possible for all $\beta\in(0,1)$.
Both conditions are comparable to the assumptions in \cite{chen2017} for an attractive interaction scaling with $\beta\in(0,\tfrac{3}{7})$\footnote{In our notation, the assumptions in \cite{chen2017} are $N^{\nu_1(\beta)}\ls\varepsilon^{-2}\ls N^{\nu_2(\beta)}$, where $\nu_1$ and $\nu_2$ are given by $\nu_1(\beta)=\frac{\beta}{1-\beta}$ and $\nu_2=\min\left\{\frac{1-\beta}{\beta},\frac{\frac35-\beta}{\beta-\frac15}\mathbbm{1}_{\beta\geq\frac15}+\infty\cdot\mathbbm{1}_{\beta<\frac15}, \frac{2\beta}{1-2\beta}, \frac{\frac78-\beta}{\beta} \right\}$. 
Note that $N^{\nu_1(\beta)}\varepsilon^2=(\frac{\varepsilon^2}{\mu})^{\frac{1}{1-\beta}}$ and $N^{\nu_2(\beta)}\varepsilon^2\leq(\frac{\varepsilon}{\mu})^\frac{2}{1-2\beta}$ as $\nu_2(\beta)\leq \frac{2\beta}{1-2\beta}$, hence these conditions are comparable to our assumptions. }.

We will use the notation $A\ls B$ to indicate that there exists a constant $C>0$ independent of $\varepsilon, N, t, \psi^{N,\varepsilon}_0,\Phi_0$ such that $A\leq CB$. The constant may depend on the quantities fixed by the model, such as $V^\perp$, $\chi$ and $\Vp$.

We consider interactions of the following type:

\begin{definition}\label{def:W}
Let $\beta\in(0,1)$ and $\eta>0$. Define the set $\mathcal{W}_{\beta,\eta}$ as the set containing all families 
$$\wb:\mathbb{N}\times(0,1)\to L^\infty(\R^3,\R), \quad (N,\varepsilon)\mapsto \wb((N,\varepsilon)),$$ 
such that for any $(N,\varepsilon)\in\mathbb{N}\times(0,1)$
$$
%\wb\in\mathcal{W}_{\beta,\eta}\; :\Longleftrightarrow \;
\begin{cases} 
	(a)\;\,\norm{\wb((N,\varepsilon))}_{L^\infty(\R^3)}\ls \left(\tfrac{N}{\varepsilon^2}\right)^{-1+3\beta},
	\vspace{0.2cm}\\\vspace{0.2cm}
	(b)\;\, \wb((N,\varepsilon)) \text{ is non-negative and spherically symmetric},\\\vspace{0.2cm}
	(c)\;\, \supp{\wb((N,\varepsilon))} \subseteq \left\{\z\in\R^3:|\z|\ls \left(\tfrac{N}{\varepsilon^2}\right)^{-\beta}\right\},\\
	(d)\;\, 	\lim\limits_{(N,\varepsilon)\rightarrow(\infty,0)}\left(\frac{N}{\varepsilon^2}\right)^\eta
	\left|b_{N,\varepsilon}((N,\varepsilon),\wb)-\bb(\wb)\right|	=0,
\end{cases}$$
where
$$b_{N,\varepsilon}((N,\varepsilon),\wb):= N\int\limits_{\R^3}\wb((N,\varepsilon),\z)\d\z\int\limits_{\R^2}|\chie(y)|^4\d y=\tfrac{N}{\varepsilon^2}\int\limits_{\R^3}\wb((N,\varepsilon),\z)\d\z\int\limits_{\R^2}|\chi(y)|^4\d y,$$
$$\bb(\wb):=\lim\limits_{(N,\varepsilon)\rightarrow(\infty,0)}b_{N,\varepsilon}((N,\varepsilon),\wb).$$ 
We will in the following abbreviate $\wb((N,\varepsilon))\equiv\wb$, $b_{N,\varepsilon}((N,\varepsilon),\wb)\equiv b_{N, \varepsilon}$ and $\bb(\wb)\equiv\bb$.
\end{definition}

Condition (d) ensures that the $(N,\varepsilon)$-dependent parameter $b_{N,\varepsilon}$ converges sufficiently fast to its limit $\bb$. Clearly, the interaction $(\tfrac{N}{\varepsilon^2})^{-1+3\beta}w((\tfrac{N}{\varepsilon^2})^\beta z)$ from the introduction is contained in this set. In this case, $b_{N,\varepsilon}=\norm{w}_{L^1(\R^3)}\int_{\R^2}|\chi(y)|^4\d y=\bb$, hence (d) is true for any $\eta>0$.

In order to formulate our main theorem, we will need two different notions of one-particle energies:
\begin{itemize}
\item The \emph{``renormalised'' energy per particle}: for $\psi\in\mathcal{D}(H_\beta(t))$,
\begin{equation}\label{E^psi}
E^\psi(t):=\tfrac{1}{N}\lr{\psi,H_\beta(t)\psi}_{L^2(\R^{3N})}-\tfrac{E_0}{\varepsilon^2},
\end{equation}
where $E_0$ denotes the lowest eigenvalue of $-\Delta_y+V^\perp(y)$. By rescaling, $\frac{E_0}{\varepsilon^2}$ is the lowest eigenvalue of $-\Delta_y+\frac{1}{\varepsilon^2}V^\perp(\frac{y}{\varepsilon^2})$.
\item The \emph{effective energy per particle}: for $\Phi\in H^2(\R)$,
\begin{equation}\label{E^Phi}
\mathcal{E}^\Phi(t):=\lr{\Phi,\left(-\tfrac{\partial^2}{\partial x^2}+\Vp(t,(x,0))+\tfrac{\bb}{2}|\Phi|^2\right)\Phi}_{L^2(\R)}.
\end{equation}
\end{itemize}
Further, we define the function $\mathfrak{e}:\R\rightarrow [1,\infty)$ by
\begin{equation}\label{def:e}
\mathfrak{e}^2(t):=1+|E^{\psi^{N,\varepsilon}_0}(0)|+|\mathcal{E}^{\Phi_0}(0)|+\int\limits_0^t \norm{\dot{\Vp}(s)}_{L^\infty(\R^3)}\d s
+\sup\limits_{\substack{i,j\in\{0,1\}\\k\in\{1,2\}}}\norm{\partial_t^i\partial_{y_k}^j\Vp(t)}_{L^\infty(\R^3)}.
\end{equation}\label{eqn:e}
This function will be of use because, by the fundamental theorem of calculus,
\begin{equation}\label{bounds:E}
\big|E^{\psi^{N,\varepsilon}(t)}(t)\big|\leq \mathfrak{e}^2(t)-1 \quad\text{ and }\quad \big|\mathcal{E}^{\Phi(t)}(t)\big|\leq \mathfrak{e}^2(t)-1 
\end{equation}
for any time $t\in\R$. Note that if the external field $\Vp$ is time-independent, $\mathfrak{e}^2(t)\ls 1$ for any $t$, hence in this case, $E^{\psi^{N,\varepsilon}(t)}(t)$  and $\mathcal{E}^{\Phi(t)}(t)$ are bounded uniformly in time.
\\

\noindent Let us now state our assumptions:
\begin{itemize}
\item[A1] \emph{Interaction.} Let the interaction $\wb\in\mathcal{W}_{\beta,\eta}$ for some $\eta>0$.
\item[A2] \emph{Confining potential.} Let $V^\perp:\R^2\rightarrow\R$ such that $-\Delta_y+V^\perp$ is self-adjoint and has a non-degenerate ground state $\chi$ with energy $E_0<\inf\sigma_\mathrm{ess}(-\Delta_y+V^\perp)$. Assume that the negative part of $V^\perp-E_0$ is bounded, i.e.~that $(V^\perp-E_0)_-\in L^\infty(\R^2)$, and that $\chi\in\mathcal{C}^1_\mathrm{b}(\R^2)$, i.e.~$\chi$ is bounded and continuously differentiable with bounded derivative. We choose $\chi$ normalised and real. 
\item[A3] \emph{External field.} Let $\Vp:\R\times\R^3\rightarrow\R$ such that for fixed $z\in\R^3$, $\Vp(\cdot,z)\in\mathcal{C}^1(\R)$. 
Further, assume that for each fixed $t\in \R$, $\Vp(t,(\cdot,0))\in H^4(\R)$, $\Vp(t,\cdot),\dot{\Vp}(t,\cdot)\in L^\infty(\R^3)\cap \mathcal{C}^1(\R^3)$ and $\nabla_y\Vp(t,\cdot),\nabla_y\dot{\Vp}(t,\cdot)\in L^\infty(\R^3)$.
%$y\cdot\nabla_y\Vp(t,\cdot)$ and $ y\cdot\nabla_y\dot{\Vp}(t,\cdot)$ are bounded for each fixed $t$.
%
\item[A4] \emph{Initial data.} Assume that the family of initial data, $\psi^{N,\varepsilon}_0\in\mathcal{D}(H_\beta(0))\cap L^2_+(\R^{3N})$ with $\norm{\psi^{N,\varepsilon}_0}^2=1$, is close to a condensate with condensate wavefunction $\phe_0=\Phi_0\chie$ for some normalised $\Phi_0\in H^2(\R)$ in the following sense: for some admissible, moderately confining sequence $(N,\varepsilon)$, it holds that
\begin{equation}\label{A4:1}
\lim\limits_{(N,\varepsilon)\rightarrow(\infty,0)}\Tr_{L^2(\R^3)}\Big|\gamma^{(1)}_{\psi^{N,\varepsilon}_0}-\ket{\Phi_0\chie}\bra{\Phi_0\chie}\Big|=0
\end{equation}
and
\begin{equation}\label{A4:2}
\lim\limits_{(N,\varepsilon)\rightarrow(\infty,0)}\left|E^{\psi^{N,\varepsilon}_0}(0)-\mathcal{E}^{\Phi_0}(0)\right|=0.
\end{equation}
\end{itemize}
\begin{remark}
\remit{
	\item Assumption A1 includes the interaction
	$\wb(\z)=\left(\tfrac{N}{\varepsilon^2}\right)^{-1+3\beta}w\left((\tfrac{N}{\varepsilon^2})^\beta\z\right)$
	for $w:\R^3\rightarrow\R$ spherically symmetric, non-negative and with $\supp{w}\subseteq \overline{B_1(0)}$.
	\item Assumption A2 is, for instance, fulfilled by a harmonic potential or by any bounded smooth potential with a bound state below the essential spectrum. Note that it is not necessary that the potential increases as $|y|\rightarrow \infty$.
	%, it is merely required to have a somewhat localised ground state. 
	The confining effect of the potential is due to the rescaling by $\varepsilon$ because the ground state of $-\Delta_y+V^\perp$ is exponentially localised \cite[Theorem 1]{griesemer2004}.
	\item The regularity condition on $\Vp(t,(\cdot,0))$ in A3  ensures the global existence of $H^2$-solutions of the NLS equation \eqref{NLS} (see Appendix \ref{appendix} and Lemma \ref{lem:Phi}). The further requirements for $\Vp,\dot{\Vp}, \nabla_y\Vp$ and $\nabla_y\dot{\Vp}$ are needed to control the one-particle energies and the interactions of bosons with the external field $\Vp$.
	\item Due to assumptions A1 -- A3, $H_\beta(t)$ is self-adjoint on $\mathcal{D}(H_\beta(t))=\mathcal{D}(H_\beta)$. As $t\mapsto \Vp(t)\in\mathcal{L}(L^2(\R^3))$ is continuous, $H_\beta(t)$ generates a strongly continuous unitary evolution on $\mathcal{D}(H_\beta)$. 
%By density, one can extend the result to initial data $\psi^{N,\varepsilon}_0\in L^2_+(\R^{3N})$.
%
	\item We assume in A4 that the system is initially given by a Bose--Einstein condensate with factorised condensate wavefunction.
	Both parts \eqref{A4:1} and \eqref{A4:2} of the assumption are standard when deriving effective evolution equations.
	For the scaling parameter $\beta=1$, it is shown in \cite{lieb2004} that the ground state of the corresponding system satisfies assumption A4. For related results without strong confinement, we refer to the review \cite{LSSY} for $\beta=1$ and to \cite{lewin2016} for $\beta<1$.\label{rem:LSSY}
}	
\end{remark}
\begin{thm}\label{thm}
Let $\beta\in(0,1)$ and assume that $\wb$, $V^\perp$ and $\Vp$ satisfy A1 -- A3. Let $\psi^{N,\varepsilon}_0$ be a family of initial data satisfying A4, let $\psi^{N,\varepsilon}(t)$ denote the solution of the $N$-body Schrödinger equation \eqref{SE} with initial datum $\psi^{N,\varepsilon}(0)=\psi^{N,\varepsilon}_0$ and let $\gamma^{(k)}_{\psi^{N,\varepsilon}(t)}$ denote its $k$-particle reduced density matrix as in \eqref{eqn:k:particle:RDM}. Then for any $t\in\R$ and $k\in\mathbb{N}$,
\begin{equation}\label{T1}
\lim\limits_{(N,\varepsilon)\rightarrow(\infty,0)}\Tr_{L^2(\R^{3k})}\Big|\gamma^{(k)}_{\psi^{N,\varepsilon}(t)}-\ket{\Phi(t)\chie}\bra{\Phi(t)\chie}^{\otimes k}\Big|=0
\end{equation}
and
\begin{equation}\label{T2}
\lim\limits_{(N,\varepsilon)\rightarrow(\infty,0)}\left|E^{\psi^{N,\varepsilon}(t)}(t)-\mathcal{E}^{\Phi(t)}(t)\right|=0,
\end{equation}
where the limits are taken along the sequence from A4. $\Phi(t)$ is the solution of the NLS equation \eqref{NLS} with initial datum $\Phi(0)=\Phi_0$ from A4, where the strength of the nonlinearity in \eqref{NLS} is given by $\bb$ from Definition \ref{def:W}, namely
\begin{equation}\label{b}
\bb=\lim\limits_{(N,\varepsilon)\rightarrow(\infty,0)}b_{N,\varepsilon}=\lim\limits_{(N,\varepsilon)\rightarrow(\infty,0)}
\tfrac{N}{\varepsilon^2} \int\limits_{\R^3}\wb(\z)\d\z  \int\limits_{\R^2}|\chi(y)|^4\d y.
\end{equation}
\end{thm}

\begin{remark}
\begin{enumerate}[label={(\alph*)}]
	\item For the specific choice $\wb(\z)=\left(\tfrac{N}{\varepsilon^2}\right)^{-1+3\beta} w\left((\tfrac{N}{\varepsilon^2})^\beta\z\right)$, we obtain the coupling parameter $\bb=\norm{w}_{L^1(\R^3)}\int_{\R^2}|\chi(y)|^4\d y$.
	\item For any fixed $t\in\R$, our proof yields an estimate of the rate of the convergence of \eqref{T1}, which is explicitly stated in Corollary \ref{cor:rate}. The rate is not uniform in time but depends on it in terms of a double exponential. Note, however, that times of order one already correspond to long times on the microscopic scale.
	\item Let us comment on the difference of our work to the result of von Keler and Teufel \cite{keler2016}, who consider $\beta\in(0,\frac13)$. The extension to $\beta\in(0,1)$  means a physically relevant improvement of the result: for $\beta<\frac13$, the problem can still be considered as a mean-field problem since the mean inter-particle distance $\varrho^{-\frac13}\sim(\frac{N}{\varepsilon^2})^{-\frac13}$ is small compared to the range of the interaction $\mu=(\frac{N}{\varepsilon^2})^{-\beta}$.
	For $\beta>\frac13$, the mean-field description breaks down and one must handle interactions which are too singular to be covered by the approach of \cite{keler2016}. We solve this by an integration by parts of the interaction, which comes at the price that one must control the kinetic energy of the $N$-particle wavefunction (Lemma \ref{lem:a_priori} and Lemma \ref{lem:E_kin}). Also, note that our admissibility condition is weaker 	than the respective condition $\varepsilon^{\frac43}/\mu\rightarrow0$ in \cite{keler2016}, which cannot be satisfied for $\beta>\frac23$. 
	\\
In \cite{keler2016}, the bosons are trapped within a quantum waveguide with non-trivial geometry. The confinement is realised by means of Dirichlet boundary conditions, which restrict the system to a tube of width $\varepsilon$ around some curve in $\R^3$.
%The reason for this choice is that the authors admit a non-trivial geometry of the waveguide, i.e.~its twisting and bending, and thus require Dirichlet boundary conditions to prevent tails of the wavefunction from overlapping different parts of the waveguide. 
In our model, the confinement is by potentials. However, our result can be easily modified to a confinement via Dirichlet boundary conditions, corresponding to a straight and untwisted quantum waveguide. The main difference in the proof is the estimate of $\gamma_b^{(1)}$ (Section \ref{subsec:gamma:b:1}): one divides the expression \eqref{eqn:Gamma(x_1)} into an integral over those ${y}$ sufficiently distant from the boundary that $\supp\wb((x,{y})-\cdot)$ is completely contained in the waveguide, and into an integral over the rest, which is easily estimated.\\
In addition to moderately confining sequences, the authors of \cite{keler2016} consider sequences $(N,\varepsilon)\rightarrow(\infty,0)$ with $\varepsilon/\mu\rightarrow0$. This is possible for $\beta\in(0,\frac12)$ and leads to $\bb=0$ in the effective equation because an essential part of the interaction is cut off such that the limiting effective equation becomes linear. We conjecture that the same effect occurs in our setup.
	\item In \cite{chen2017}, Chen and Holmer study attractive interactions, i.e.~$\int_{\R^3}\wb(\z)\d\z\leq 0$. In distinction from that work, we exclusively consider repulsive interactions with $\wb\geq0$. However, as the condition $\wb\geq0$ seems to be crucial only to the proofs of Lemma \ref{lem:a_priori} and Lemma \ref{lem:E_kin}, it is likely that our result can be extended to include repulsive interactions with a certain negative part.
\end{enumerate}
\end{remark}

\section{Proof of the main theorem}\label{sec:proof}
To prove Theorem \ref{thm}, we need to show that the expressions in \eqref{T1} and \eqref{T2} vanish in the limit $(N,\varepsilon)\rightarrow(\infty,0)$, given suitable initial data. Instead of estimating these differences directly, we adhere to the idea by Pickl \cite{pickl2008,pickl2015_arxiv,pickl2010,pickl2011,pickl2015} to define a functional $\alpha_\xi(\psi^{N,\varepsilon}(t),\phe(t))$ which provides a measure of the part of the $N$-particle wavefunction $\psi^{N,\varepsilon}$ that has not condensed into the single-particle orbital $\phe$.
The functional is chosen such that $\alpha_\xi(\psi^{N,\varepsilon}(t),\phe(t))\rightarrow0$ is equivalent to \eqref{T1} and \eqref{T2}.
We follow in general \cite{pickl2015}. However, the strongly asymmetric confinement requires a nontrivial modification of the formalism to treat the dimensional reduction and the more singular scaling of the interaction.
For the construction of $\alpha_\xi$, we need the following projections:

\begin{definition}\label{def:p}
Let $\phe(t)=\Phi(t)\chie$, where $\Phi(t)$ is the solution of the NLS equation \eqref{NLS} with initial datum $\Phi_0$ from A4 and with $\chie$ as in \eqref{chie}. Let
$$p:=\ket{\phe(t)}\bra{\phe(t)},$$
where we have dropped the time dependence of $p$ in the notation. For $i\in\{1,\dots,N\}$, define the projection operators on $L^2(\R^{3N})$
$$p_j:=\underbrace{\mathbbm{1}\otimes\cdots\otimes\mathbbm{1}}_{j-1}\otimes\, p\otimes \underbrace{\mathbbm{1}\otimes\cdots\otimes\cdots\mathbbm{1}}_{N-j} \quad\text{and}\quad q_j:=\mathbbm{1}-p_j.$$
Further, define the orthogonal projections on $L^2(\R^3)$ 
\begin{align*}
\pp&:=\ket{\Phi(t)}\bra{\Phi(t)}\otimes\mathbbm{1}_{L^2(\R^2)},  &  \qp&:=\mathbbm{1}_{L^2(\R^3)}-\pp,\\
\pc&:=\mathbbm{1}_{L^2(\R)}\otimes\ket{\chie}\bra{\chie}, &  \qc&:=\mathbbm{1}_{L^2(\R^3)}-\pc,
\end{align*}
and define $\pp_j$, $\qp_j$, $\pc_j$ and $\qc_j$ on $L^2(\R^{3N})$ analogously to $p_j$ and $q_j$.
Finally, for $0\leq k\leq N$, define the many-body projections
$$ P_k=\big(q_1\cdots q_k p_{k+1}\cdots p_N\big)_\mathrm{sym}:=\sum\limits_{\substack{J\subseteq\{1,\dots,N\}\\|J|=k}}\prod\limits_{j\in J}q_j\prod\limits_{l\notin J}p_l $$
and $P_k=0$ for $k<0$ and $k>N$.
\end{definition}

We will write $p_j=\ket{\phe(t,z_j)}\bra{\phe(t,z_j)}$, $\pp_j=\ket{\Phi(t,x_j)}\bra{\Phi(t,x_j)}$ and $\pc_j=\ket{\chie(y_j)}\bra{\chie(y_j)}$. Some useful identities of the projections are listed in the following corollary:

\begin{cor}\label{cor:p}
For $0\leq k\leq N$ and $1\leq j\leq N$, it holds that
\lemit{
\item $\sum\limits_{k=0}^N P_k=\mathbbm{1},$  \quad $\sum\limits_{j=1}^N q_j P_k=kP_k,$\label{cor:p:1}
\item $p_j=\pp_j\pc_j$, \quad $p_j^{\Phi/\chie} p_j= p_j$, \quad $q_j^{\Phi/\chie} q_j= q_j^{\Phi/\chie}$, \quad $p^{\Phi/\chie}_j q_j=p_j^{\Phi/\chie} q_j^{\chie/\Phi},$ \quad $q_j^{\Phi/\chie}p_j=0,$\label{cor:p:2}
\item $q_j=\qp_j \pc_j+\pp_j\qc_j+\qp_j\qc_j=\qc_j+\qp_j\pc_j=\qp_j+\pp_j \qc_j$.\label{cor:p:3}
}
\end{cor}
\begin{proof}
The first identity in (a) is due to the relation $p_j+q_j=\mathbbm{1}$. The second identity follows from the fact that
$$\sum\limits_{j=1}^Nq_j=\sum\limits_{j=1}^Nq_j\sum\limits_{k=0}^N P_k=\sum\limits_{k=0}^N\sum\limits_{j=1}^Nq_jP_k=\sum\limits_{k=0}^NkP_k$$
together with $P_kP_{k'}=\delta_{k,k'}P_k$. While part (b) is an immediate consequence of Definition~\ref{def:p}, part (c) is implied by
$q=1-p=(\pp+\qp)(\pc+\qc)-\pp\pc=\pp\qc+\qp\pc+\qp\qc.$
\end{proof}

\begin{definition}
For any function $f: \mathbb{N}_0\rightarrow\R_0^+$, define the operator $\hat{f}\in\mathcal{L}\left(L^2(\R^{3N})\right)$ by $$\hat{f}:= \sum\limits_{k=0}^N f(k)P_k$$
and, for any $d\in\mathbb{Z}$, the shifted operator $\hat{f}_d\in\mathcal{L}\left(L^2(\R^{3N})\right)$ by
$$\hat{f}_d:=\sum\limits_{j=-d}^{N-d} f(j+d)P_j.$$
We will in particular need the weight function $n$ defined by $n(k):=\sqrt{\tfrac{k}{N}}.$
\end{definition}

We will exclusively use the symbol $\,\hat{\cdot}\,$ to denote such weighted many-body operators. Besides, we will in the following use the abbreviations 
$$\llr{\cdot,\cdot}:=\lr{\cdot,\cdot}_{L^2(\R^{3N})},\quad \norm{\cdot}:=\norm{\cdot}_{L^2(\R^{3N})}
\quad\text{and}\quad\onorm{\cdot}:=\norm{\cdot}_{\mathcal{L}(L^2(\R^{3N}))}.$$

\begin{definition}
Define the functional $\alpha_f: L^2(\R^{3N})\times L^2(\R^{3})\rightarrow \R$ by
$$\alpha_f(\psi,\phe(t)):=\llr{\psi,\hat{f}\psi}=\sum\limits_{k=0}^Nf(k)\llr{\psi,P_k\psi}.$$
\end{definition}

The $\phe$-dependence of $\alpha_f$ is due to the $\phe$-dependence of the projectors $P_k$. As the operators $P_k$ project onto states with exactly $k$ particles outside the condensate, $\alpha_f$ is a measure of the relative number of such particles in the state $\psi$. We choose the weight $f$ increasing and $f(0)\approx0$, hence those parts of $\psi$ with a larger ``distance'' to the condensate contribute more to $\alpha_f(\psi,\phe)$. On the other hand, $P_0\psi$ --- the state where all particles are condensed into $\phe$ --- contributes hardly anything. The weight $\hat{n}$ is in particular distinguished because for any symmetric wavefunction $\psi\in L^2_+(\R^{3N})$,
$$\alpha_{n^2}(\psi,\phe(t))=\sum\limits_{k=0}^N\tfrac{k}{N}\llr{\psi,P_k\psi}
=\sum\limits_{k=0}^N\sum\limits_{j=1}^N\tfrac1N\llr{\psi,q_jP_k\psi}=\norm{q_1\psi}^2$$
by Corollary \ref{cor:p:1}.

\begin{lem}\label{lem:equivalences}
Let $\psi^N\in L^2_+(\R^{3N})$ be a sequence of normalised $N$-particle wavefunctions and let $\gamma_N^{(k)}$ be the sequence of corresponding $k$-particle reduced density matrices for some fixed $k\in\mathbb{N}$. Let $t\in\R$. Then the following statements are equivalent:
\lemit{
	\item $\lim\limits_{N\rightarrow\infty}\alpha_{n^a}(\psi^N,\phe(t))=0$ for some $a>0$,
	\item $\lim\limits_{N\rightarrow\infty}\alpha_{n^a}(\psi^N,\phe(t))=0$ for any $a>0$,
	\item $\lim\limits_{N\rightarrow\infty}\norm{\gamma^{(k)}_N-\ket{\phe(t)}\bra{\phe(t)}^{\otimes k}}_{\mathcal{L}(L^2(\R^{3k}))}=0 $ for all $ k\in\mathbb{N}$,
	\item $\lim\limits_{N\rightarrow\infty}\Tr_{L^2(\R^{3k})}\left|\gamma^{(k)}_N-\ket{\phe(t)}\bra{\phe(t)}^{\otimes k}\right|=0 $ for all $ k\in\mathbb{N}$,
	\item $\lim\limits_{N\rightarrow\infty}\Tr_{L^2(\R^{3})}\left|\gamma^{(1)}_N-\ket{\phe(t)}\bra{\phe(t)}\right|=0.$
}
\end{lem}
\noindent For the proof of this lemma, we refer to \cite[Lemma 3.1]{keler2016} and to corresponding results in \cite{knowles2010,pickl2011}.
We will in the following choose the weight function $m:\mathbb{N}_0\rightarrow\R^+_0$ with
$$ m(k):=\begin{cases}
	n(k) & \text{for } k\geq N^{1-2\xi},\\
	\tfrac12\left(N^{-1+\xi}k+N^{-\xi}\right) & \text{else}\end{cases}$$ 
for some $\xi\in(0,\frac12)$, i.e.~$m$ equals $n$ with a smooth cut-off to soften the singularity of $\frac{\d n}{\d k}$ for small $k$. Clearly, $n(k)\leq m(k)\leq n(k)+\tfrac12N^{-\xi}$ for all $k\geq 0$ and $\xi\in(0,\frac12)$, hence $\alpha_m(\psi,\phe(t))\rightarrow0$ is equivalent to $\alpha_n(\psi,\phe(t))\rightarrow0$ and thus to all cases in Lemma~\ref{lem:equivalences} for any choice of the parameter $\xi$. 
For the actual proof, we will consider a modified version of this functional, namely
\begin{equation}\label{alpha}
\alpha_\xi(t):=\alpha_m(\psi^{N,\varepsilon}(t),\phe(t))+\big|E^{\psi^{N,\varepsilon}(t)}(t)-\mathcal{E}^{\Phi(t)}(t)\big|.
\end{equation}
The convergence of $\alpha_\xi(t)$ to zero is equivalent to \eqref{T1} and \eqref{T2}. Conversely, \eqref{A4:1} and \eqref{A4:2} imply $\alpha_\xi(0)\rightarrow 0$ as $(N,\varepsilon)\rightarrow(\infty,0)$. The main idea of the proof is therefore to derive a bound for $|\tfrac{\d}{\d t}\alpha_\xi(t)|$ (Propositions \ref{prop:alpha} and \ref{prop:gamma}), from which one obtains an estimate for $\alpha_\xi(t)$ by Grönwall's inequality. The propositions will be proven in Sections \ref{subsec:prop:alpha} and \ref{subsec:prop:gamma}. 
The estimate of the rate of the convergence of $\alpha_\xi(t)$ gained from this procedure translates to a rate for the reduced density matrices:

\begin{lem}\label{lem:comparison:rates}
For $\alpha_\xi(t)$ as in \eqref{alpha}, it holds that
\begin{align*}
&\Tr\left|\gamma_{\psi^{N,\varepsilon}(t)}^{(1)}-\ket{\phe(t)}\bra{\phe(t)}\right|\leq \sqrt{8\alpha_\xi(t)},\\
&\alpha_\xi(t)\leq \left|E^{\psi^{N,\varepsilon}(t)}(t)-\mathcal{E}^{\Phi(t)}(t)\right|+\sqrt{\Tr\left|\gamma_{\psi^{N,\varepsilon}(t)}^{(1)}-\ket{\phe(t)}\bra{\phe(t)}\right|}+\tfrac12 N^{-\xi}.
\end{align*}
\end{lem}

\begin{proof}
Let us abbreviate $\psi^{N,\varepsilon}(t)\equiv\psi$ and drop all time dependencies. 
\cite[Lemma 2.3]{knowles2010} implies 
$$\llr{\psi,\hat{n}^2\psi}\leq \Tr\big|\gamma_\psi^{(1)}-\ket{\phe}\bra{\phe}\big|\leq \sqrt{8\llr{\psi,\hat{n}^2\psi}}.$$
The first inequality is thus immediately clear as $n(k)^2\leq n(k)\leq m(k)$. For the second inequality, recall that $m(k)\leq n(k)+\frac12 N^{-\xi}$, hence
$$\llr{\psi,\hat{m}\psi}\leq \norm{\psi}\norm{\hat{n}\psi}+\tfrac12 N^{-\xi}
\leq\sqrt{\llr{\psi,\hat{n}^2\psi}}+\tfrac12 N^{-\xi}\leq \sqrt{\Tr\big|\gamma_\psi^{(1)}-\ket{\phe}\bra{\phe}\big|}+\tfrac12 N^{-\xi}.$$
\end{proof}

\begin{prop}\label{prop:alpha}
Under assumptions A1 -- A4,
\begin{align*}
\left|\tfrac{\d}{\d t}\alpha_\xi(t) \right|&\leq \big|\gamma_a(t)\big|+|\gamma_b(t)|\\
&\leq\big|\gamma_a(t)\big|+\big|\gamma_b^{(1)}(t)\big|+\big|\gamma_b^{(2)}(t)\big|+\big|\gamma_b^{(3)}(t)\big|
\end{align*}
for almost every $t\in\R$,
where
\begin{align}
	\gamma_a(t):=&\Big|\llr{\psi^{N,\varepsilon}(t),\dot{\Vp}(t,\z_1)\psi^{N,\varepsilon}(t)}-\lr{\Phi(t),\dot{\Vp}(t,(x,0))\Phi(t)}_{L^2(\R)}\Big|\label{gamma_a:1}\\
	&-2N\Im\llr{\psi^{N,\varepsilon}(t),q_1\hat{m}^a_{-1}\big(\Vp(t,\z_1)-\Vp(t,(x_1,0))\big)p_1\psi^{N,\varepsilon}(t)}, \label{gamma_a:2}\\\nonumber\\
	\gamma_b(t):=&- N(N-1)\Im\llr{\psi^{N,\varepsilon}(t),Z_\beta^{(12)}\hat{m}\psi^{N,\varepsilon}(t)}
\label{gamma_b},\\\nonumber\\
	\gamma_b^{(1)}(t):=&-2N(N-1)\Im\llr{\psi^{N,\varepsilon}(t),\qp_1\hat{m}^a_{-1}\pc_1p_2Z_\beta^{(12)}p_1p_2\psi^{N,\varepsilon}(t)},\label{gamma_b_1}\\
	\nonumber\\\
	\gamma_b^{(2)}(t):=&-N(N-1)\Im\llr{\qc_1\psi^{N,\varepsilon}(t),\left(2p_2\hat{m}^a_{-1}+q_2(1+\pc_2)\hat{m}^b_{-2}\right)\wbot p_1p_2\psi^{N,\varepsilon}(t)} 
		\label{gamma_b_2:1}\\
		&-2N(N-1)\Im\llr{\psi^{N,\varepsilon}(t),(\qc_1q_2+\qp_1\pc_1\qc_2)\hat{m}^a_{-1}\wbot p_1q_2\psi^{N,\varepsilon}(t)}\label{gamma_b_2:2}\\
		&-2N(N-1)\Im\llr{\psi^{N,\varepsilon}(t),\qp_1\qp_2\hat{m}^a_{-1}\pc_1\pc_2\wbot p_1\qc_2\psi^{N,\varepsilon}(t)},\label{gamma_b_2:3}\\
	\nonumber\\\
	\gamma_b^{(3)}(t) := &-N(N-1)\Im\llr{\psi^{N,\varepsilon}(t),\qp_1\qp_2\hat{m}^b_{-2}\pc_1\pc_2\wbot p_1p_2\psi^{N,\varepsilon}(t)}  \label{gamma_b_3:1}\\
		&-2N(N-1)\Im\llr{\psi^{N,\varepsilon}(t),\qp_1\qp_2\hat{m}^a_{-1}\pc_1\pc_2\wbot p_1\pc_2\qp_2\psi^{N,\varepsilon}(t)}\label{gamma_b_3:2}\\
		&+2N\bb\Im\llr{\psi^{N,\varepsilon}(t),q_1q_2\hat{m}^a_{-1}|\Phi(t,x_1)|^2p_1q_2\psi^{N,\varepsilon}(t)}.\label{gamma_b_3:3}
\end{align}
Here,
$$\wbot :=\wb(\z_1-\z_2) \quad \text{ and }\quad Z_\beta^{(12)}:=\wbot -\tfrac{\bb}{N-1}\left(|\Phi(t,x_1)|^2+|\Phi(t,x_2)|^2\right)$$
and $\hat{m}^a$, $\hat{m}^b$ denote the many-body operators corresponding to the weight functions
$$m^a(k):=m(k)-m(k+1) \quad \text{ and }\quad m^b(k):=m(k)-m(k+2).$$
\end{prop}

\noindent The first term, $\gamma_a$, merely contains one-body contributions, i.e.~interactions between the bosons and the external field $\Vp$, and is therefore the easiest to estimate. Note that \eqref{gamma_a:1} is small only if the system is in a state $\psi^{N,\varepsilon}$ close to the condensate with condensate wavefunction $\phe=\Phi\chie$ (see Lemma \ref{lem:psi-Phi}). The term $\gamma_b$ handles the two-body contributions, i.e.~interactions among bosons.
The expressions $\gamma_b^{(1)}$ and $\gamma_b^{(3)}$ contain the quasi one-dimensional interaction $\overline{w}(x_1-x_2)$ defined by $\pc_1\pc_2\wb(\z_1-\z_2)\pc_1\pc_2=:\overline{w}(x_1-x_2)\pc_1\pc_2$ (see Definition \ref{def:1d:eff:int}), where the transverse degrees of freedom are integrated out. These terms are comparable to the corresponding three-dimensional terms in \cite{pickl2015}. $\gamma_b^{(2)}$ has no equivalent in the situation without strong confinement as it collects the remainders that arise upon approximating the three-dimensional interaction $\wb$ with the quasi one-dimensional interaction $\overline{w}$. 

$\gamma_b^{(1)}$ is physically most relevant because it depends on the difference between the quasi one-dimensional interaction $\overline{w}$ and the one-dimensional effective potential $\bb|\Phi(t)|^2$. In other words, this term is small if and only if \eqref{NLS} is the right effective equation, in particular with the correct coupling parameter $\bb$. Note that for this term it is crucial that the sequence $(N,\varepsilon)$ is moderately confining, i.e.~that $\mu/\varepsilon\rightarrow0$. 

For $\gamma_b^{(2)}$ to be small, we require in particular the admissibility of the sequence $(N,\varepsilon)$, i.e.~that $\varepsilon^2/\mu\rightarrow0$. The other key tool for the estimate is the observation that due to the strong confinement, it is unlikely that a particle is excited in the transverse directions. This implies in particular that $\norm{\qc_1\psi^{N,\varepsilon}(t)}=\mathcal{O}(\varepsilon)$ (Lemma \ref{lem:a_priori}).

The estimate of $\gamma_b^{(3)}$ relies on a bound for the kinetic energy of the part of $\psi^{N,\varepsilon}(t)$ with at least one particle orthogonal to $\Phi(t)$, i.e.~a bound for $\norm{\partial_{x_1}\qp_1\psi^{N,\varepsilon}(t)}$ (Lemma \ref{lem:E_kin}). The proof of this bound again involves the splitting of the interaction $\wb$ into a quasi one-dimensional part $\overline{w}$ and remainders. Hence for $\gamma_b^{(3)}$ to be small, we require both moderate confinement and the admissibility of the sequence $(N,\varepsilon)$.
The last line \eqref{gamma_b_3:3} is a remainder which is easily controlled.

\begin{prop}\label{prop:gamma}
Let $\mu$ be sufficiently small. Under assumptions A1 -- A4, $\gamma_a$ to $\gamma_b^{(3)}$ from Proposition~\ref{prop:alpha} are bounded by
\begin{align*}
	\big|\gamma_a(t)\big| & \ls \left(\alpha_\xi(t)+\varepsilon\right)\mathfrak{e}^3(t)\vphantom{\Bigg(},\\
	\big|\gamma_b^{(1)}(t)\big| & \ls \left(\tfrac{\mu}{\varepsilon}+N^{-1}+(\tfrac{N}{\varepsilon^2})^{-\eta}\right)\mathfrak{e}^2(t),\\
	\big|\gamma_b^{(2)}(t)\big| & \ls \left(\tfrac{\varepsilon^2}{\mu}\right)^\frac12\mathfrak{e}^3(t) ,\\
	\big|\gamma_b^{(3)}(t)\big| & \ls \Big(\alpha_\xi(t)+\tfrac{\mu}{\varepsilon}
	+\left(\tfrac{\varepsilon^2}{\mu}\right)^\frac12+N^{-\frac{\beta_1}{2}}+N^{-1+\beta_1+\xi}+(\tfrac{N}{\varepsilon^2})^{-\eta}\Big)\mathfrak{e}(t)\exp\left\{\mathfrak{e}^2(t)+\int_0^t\mathfrak{e}^2(s)\d s\right\}
\end{align*}
for any $\xi\in(0,\frac{\beta}{4}]$, any $\beta_1\in(0,\beta]$ and with $\eta$ from Definition \ref{def:W} and $\mathfrak{e}(t)$ as in \eqref{def:e}.
\end{prop}

\noindent 
The estimate of $\gamma_b^{(1)}$ is essentially the same as in the case $\beta\in(0,\frac13)$ in \cite{keler2016}. $\gamma_a$ must be treated in a different way because the confinement is by a potential and not via Dirichlet boundary conditions. For the terms $\gamma_b^{(2)}$ and $\gamma_b^{(3)}$, the argument from \cite{keler2016} does not work because the interaction becomes too singular for  $\beta>\frac13$. To cope with this, we follow an idea from \cite{pickl2015}: we write the interaction as $\wb=\Delta\he$ for some function $\he$ and integrate by parts. $\na\he$ is less singular, and the expressions resulting from $\na$ acting on $\psi^{N,\varepsilon}(t)$ can be controlled with Lemma \ref{lem:a_priori} (or the refined version, Lemma \ref{lem:E_kin}). 

Our strategy differs from \cite{pickl2015} in a relevant point: in \cite{pickl2015}, the interaction $\wb$ is approximated by a potential $U_\bo$ with a softer scaling behaviour ($\bo<\frac13$). The author first proves bounds for $\beta<\frac13$, the second step it to estimate the contribution from the difference $\wb-U_\bo$ using integration by parts.
Instead of these two steps, we define $\he$ as the solution of $\Delta\he=\wb$ on a ball with Dirichlet boundary conditions and integrate by parts on the ball. To prevent the emergence of boundary terms, we use smooth step functions whose derivatives can be controlled. This mathematical trick enables us to avoid the separate estimate for $\beta<\frac13$. 

The control of the kinetic energy (Lemma \ref{lem:E_kin}) required for the integration by parts in $\gamma_b^{(3)}$ is also different from the corresponding Lemma 5.2 in \cite{pickl2015}. 
Instead of following that path, we extend ideas from \cite[Lemma 4.7]{keler2016} and \cite[Lemma 4.6]{pickl2015_arxiv} and estimate the part of the kinetic energy in the free direction.
Besides, we use with Lemma \ref{lem:GammaLambda} a slightly sharpened version of \cite[Lemma 4.3]{pickl2015}. 
\\

\noindent\emph{Proof of Theorem \ref{thm}.}
From Propositions \ref{prop:alpha} and \ref{prop:gamma}, we gather that for sufficiently small $\mu$, 
$$\left|\tfrac{\d}{\d t}\alpha_\xi(t)\right|\ls C(t)
\big(\alpha_\xi(t)+R_{\xi,\bo,\eta}(N,\varepsilon)\big)$$
for almost every $t\in\R$, where
\begin{align}
C(t)&:=\mathfrak{e}(t)\exp\left\{\mathfrak{e}^2(t)+\int_0^t\mathfrak{e}^2(s)\d s\right\},\label{C}\\
R_{\xi,\bo,\eta}(N,\varepsilon)&:=\tfrac{\mu}{\varepsilon}+\left(\tfrac{\varepsilon^2}{\mu}\right)^\frac12+N^{-\frac{\bo}{2}}+N^{-1+\bo+\xi}+(\tfrac{N}{\varepsilon^2})^{-\eta}.\nonumber
\end{align}
The differential version of Grönwall's inequality yields
$$ \alpha_\xi(t)\leq\big(\alpha_\xi(0)+R_{\xi,\bo,\eta}(N,\varepsilon)\big)\exp\left\{2\int_0^tC(s)\d s\right\}$$
for all $t\in\R$.
Due to assumption A4 and by Lemma \ref{lem:equivalences}, $\lim_{(N,\varepsilon)\rightarrow(\infty,0)}\alpha_\xi(0)=0$ and
$R_{\xi,\bo,\eta}(N,\varepsilon)$ vanishes in the limit $(N,\varepsilon)\rightarrow(\infty,0)$ for $\bo\in(0,\beta]$ and $\xi\in(0,\frac{\beta}{4}]$, $\xi<1-\bo$, because the sequence $(N,\varepsilon)$ is by assumption A4 admissible and moderately confining.  Again by Lemma \ref{lem:equivalences}, this implies \eqref{T1} and \eqref{T2} for any $t\in\R$. 
\qed

\begin{cor}\label{cor:rate}
Let $t\in\R$. Then
\begin{align*}
\Tr\left|\gamma^{(1)}_{\psi^{N,\varepsilon}(t)}-\ket{\phe(t)}\bra{\phe(t)}\right|\ls
\left(A(0)+\tfrac{\mu}{\varepsilon}+\left(\tfrac{\varepsilon^2}{\mu}\right)^{\frac12}
+N^{-\frac{\beta}{4}}+(\tfrac{N}{\varepsilon^2})^{-\eta}\right)^\frac12 
\exp\left\{\int_0^tC(s)\d s\right\}
\end{align*}
for $C(t)$ as in \eqref{C} and where
$$A(0):=\left|E^{\psi^{N,\varepsilon}_0}(0)-\mathcal{E}^{\Phi_0}(0)\right|+\sqrt{\Tr\Big|\gamma^{(1)}_{\psi_0^{N,\varepsilon}}-\ket{\phe_0}\bra{\phe_0}\Big|}.$$
\end{cor}

\begin{proof}
This follows from Lemma \ref{lem:comparison:rates} after optimisation over $\xi$ and $\bo$.
\end{proof}

\section{Proofs of the propositions}\label{sec:prop} 
\subsection{Preliminaries}
In this section, we prove several lemmata which are needed for the proofs of the propositions.
The first ones establish several properties of the weighted operators $\hat{f}$, Lemma \ref{lem:psi-Phi} 
and Lemma \ref{lem:GammaLambda:and:O_12} contain some useful estimates for scalar products, and the remainder of the section covers properties of the condensate wavefunction $\phe(t)$.
In the following, we will always assume that assumptions A1 -- A4 are satisfied.

\begin{lem}\label{lem:l}
Let $f:\mathbb{N}_0\rightarrow\mathbb{R}_0^+$, $d\in\mathbb{Z}$ and define $$\hat{l}:=N \max \{\hat{m}^a_{-1},\hat{m}^b_{-2}\},$$ where the $\max$ is to be understood in the sense of inequalities between operators, i.e.~$\hat{l}=N\hat{m}^a_{-1}$ if $\hat{m}^a_{-1}-\hat{m}^b_{-2}$ is a positive operator and vice versa. Then
\lemit{
 	\item $\onorm{\hat{f}}=\onorm{\hat{f}_d}=\onorm{\hat{f}^\frac{1}{2}}^2=\sup\limits_{0\leq k\leq N}f(k)$,\label{lem:l:1} 
	\item $\onorm{\hat{l}\,\hat{n}}\ls 1$, \quad $\onorm{\hat{l}}\leq N^{\xi}.$ \label{lem:l:2} 
}
\end{lem}

\begin{proof} 
Part (a) is obvious. For part (b), note that
$$\hat{m}^a_{-1}\hat{n}=\sum\limits_{k=1}^N\left(m(k-1)-m(k)\right)n(k)P_k, \qquad \hat{m}^b_{-2}=\sum\limits_{k=2}^N\left(m(k-2)-m(k)\right)n(k)P_k.
$$
The derivative of $m$ with respect to $k$ is given by
\begin{equation*}
m'(k)\equiv\tfrac{\d}{\d k}m(k)= \left\{
\begin{array}{ll}
      \frac{1}{2\sqrt{kN}} =\frac12N^{-1}n(k)^{-1}& \text{for}\;\; k\geq N^{1-2\xi},\vphantom{\bigg(}\\
      \frac12N^{-1+\xi} & \text{else}.\\
\end{array} 
\right.
\end{equation*}
By the mean value theorem, $|m(k)-m(k-j)|=j|m'(\kappa)| $ for $j\in\{1,2\}$ and $\kappa\in (k-j,k)$. For $\kappa\geq N^{1-2\xi}$, $|m'(\kappa)|=\frac12 N^{-1}n(\kappa)^{-1}$.
For $\kappa<N^{1-2\xi}$, we obtain $|m'(\kappa)|=\frac12 N^{-1+\xi}<\frac12\frac{1}{\sqrt{\kappa N}}=\frac12N^{-1}n(\kappa)^{-1}$. Consequently,
$$\sum\limits_{k=j}^N\big|m(k-j)-m(k)\big|n(k)P_k\leq \tfrac{1}{2}N^{-1}j\sum\limits_{k=j}^N\sqrt{\tfrac{k}{\kappa}}P_k \ls N^{-1}\mathbbm{1}$$
in the sense of operators. This proves the first part of (b). 
For the second identity, observe that $|m'(k)|\leq \frac12 N^{-1+\xi}$ uniformly in $k\geq 0$.
\end{proof}

%The next lemma establishes some useful commutation relations.

\begin{lem}\label{lem:commutators}
Let $f,g:\mathbb{N}_0\rightarrow\mathbb{R}_0^+$ be any weights and $i,j\in\{1\mydots N\}$. 
\lemit{
	\item	\label{lem:commutators:1} For $k\in\{0\mydots N\}$,
			\begin{equation*}
			\hat{f}\,\hat{g}=\hat{fg}=\hat{g}\hat{f},\qquad\hat{f}p_j=p_j\hat{f},\qquad\hat{f}q_j=q_j\hat{f}, \qquad\hat{f}P_k=P_k\hat{f}.
			\end{equation*} 	
	\item	\label{lem:commutators:2}
			Define $Q_0:=p_j$, $Q_1:=q_j$, $\tilde{Q}_0:=p_ip_j$, $\tilde{Q}_1\in\{p_iq_j,q_ip_j\}$ and $\tilde{Q}_2:=q_iq_j$. 
			Let $S_j$ be an operator acting only on factor $j$ in the tensor product and $T_{ij}$ acting only on $i$ and $j$.
			Then for $\mu,\nu\in\{0,1,2\}$	
			$$ Q_\mu\hat{f}S_jQ_\nu=Q_\mu S_j\hat{f}_{\mu-\nu}Q_\nu \quad \text{ and } \quad
			\tilde{Q}_\mu\hat{f}T_{ij}\tilde{Q}_\nu=\tilde{Q}_\mu T_{ij}\hat{f}_{\mu-\nu}\tilde{Q}_\nu.$$
	\item	\label{lem:commutators:4} 
			Let $S_{x_j}$ be an operator acting only on the $x$-component of factor $j$. Then
			$$\qp_j\hat{f}S_{x_j}\pp_j=\qp_jS_{x_j}(\hat{f}\qc_j+\hat{f}_1\pc_j)\pp_j 
			\quad \text{ and }\quad
			\qp_j\hat{f}S_{x_j}\qp_j=\qp_jS_{x_j}\hat{f}\qp_j.$$
	\item	\label{lem:commutators:5}
			\begin{equation*}
			[T_{12},\hat{f}]=[T_{12},p_1p_2(\hat{f}-\hat{f}_2)+(p_1q_2+q_1p_2)(\hat{f}-\hat{f}_1)].
			\end{equation*}
}
\end{lem}

We will apply parts (b) and (c) to unbounded operators, for instance to $S_j\equiv \nabla_j$ and $S_{x_j}\equiv \partial_{x_j}$. In this case, the respective equality holds on the intersection of the domains of the operators on both sides of the equation.

\begin{proof}
Part (a) follows immediately from $P_kP_l=\delta_{k,l}P_k$. For assertion (b), note that for $j=1$,
\begin{equation*}\begin{split}
Q_\mu P_k S_1Q_\nu&
= Q_\mu\Bigg(\sum\limits_{\substack{J\subseteq\{2,\dots,N\}\\|J|=k-\mu}}\prod\limits_{j\in J}q_j\prod\limits_{l\notin J}p_l\Bigg) S_1Q_\nu\\
&=Q_\mu S_1\Bigg(\sum\limits_{\substack{J\subseteq\{2,\dots,N\}\\|J|=k-\mu}}\prod\limits_{j\in J}q_j\prod\limits_{l\notin J}p_l\Bigg) Q_\nu
=Q_\mu S_1 P_{k-\mu+\nu}Q_\nu,
\end{split}\end{equation*}
hence
\begin{equation*}\begin{split}
Q_\mu \hat{f}S_1Q_\nu
&=Q_\mu S_1 \left(\sum\limits_{k=-(\mu-\nu)}^{N-(\mu-\nu)}f(k+\mu-\nu)P_{k}\right) Q_\nu =Q_\mu S_1\hat{f}_{\mu-\nu}Q_\nu.
\end{split}\end{equation*}
Assertion (c) is a consequence of part (b) and Corollary \ref{cor:p:2}, for example
$$\qp_j\hat{f}S_{x_j}\pp_j=\qp_j\left(q_j\hat{f}S_{x_j}(p_j+q_j)\right)\pp_j=\qp_j S_{x_j}(\hat{f}_1\pc_j+\hat{f}\qc_j)\pp_j.$$
Finally, observe that
\begin{equation*}\begin{split}
[T_{12},p_1p_2(\hat{f}-\hat{f}_2)&+(p_1q_2+q_1p_2)(\hat{f}-\hat{f}_1)]=[T_{12},\hat{f}]-[T_{12},q_1q_2\hat{f}+(p_1q_2+q_1p_2)\hat{f}_1+p_1p_2\hat{f}_2].
\end{split}\end{equation*}
The second commutator equals zero, which can be seen by inserting $1=p_1p_2+(p_1q_2+q_1p_2)+q_1q_2$ in front of the commutator and applying part (d).
\end{proof}

%Some further properties of the weighted many-body operators are collected in the following lemma.

\begin{lem}\label{lem:derivative_m}
Let $f:\mathbb{N}_0\rightarrow\R^+_0$. Then
\lemit{
	\item	\label{lem:derivative_m:1}
			$P_k,\;\hat{f}\in \mathcal{C}^1\big(\R,\mathcal{L}\left(L^2(\R^{3N})\right)\big)$ for $0\leq k\leq N$,
	\item	\label{lem:derivative_m:2}
			$\left[-\Delta_{y_j}+\frac{1}{\varepsilon^2}V^\perp(\tfrac{y_j}{\varepsilon}),\hat{f}\right]=0$ for $1\leq j \leq N$,
	\item	\label{lem:derivative_m:3}
			$\tfrac{\d}{\d t}\hat{f}=i\Big[\hat{f},\sum\limits_{j=1}^N h_j(t)\Big],$\\
			where $h_j(t)$ denotes the one-particle operator corresponding to $h(t)$ from \eqref{NLS} acting on the $j$\textsuperscript{th} factor in $L^2(\R^{3N})$.
}
\end{lem}

\begin{proof}
Part (a) is clear as $\phe\in \mathcal{C}^1\left(\R,L^2(\R^3)\right)$. Assertion (b) is due to the fact that $-\Delta_{y_j}+\frac{1}{\varepsilon^2}V^\perp(\tfrac{y_j}{\varepsilon})$ commutes with its spectral projection $\pc_j$. For the last part, note that
\begin{equation*}\begin{split}
\tfrac{\d}{\d t}p=\tfrac{\d}{\d t}\ket{\Phi(t)\chie}\bra{\Phi(t)\chie}=i\left[\ket{\Phi(t)\chie}\bra{\Phi(t)\chie},h(t)\right]=i[p,h(t)] \quad \text{and}\quad
\tfrac{\d}{\d t}q=i[q,h(t)]
\end{split}\end{equation*}
as $\Phi(t)$ is a solution of~\eqref{NLS}.
\end{proof}

We will consider functions which are symmetric only in the variables of a subset of $\{1\mydots N\}$, for instance the expressions $q_1\psi$ and $\wbot \psi$ for $\psi\in L^2_+(\R^{3N})$. 

\begin{definition}\label{def:H_M}
Let $\mathcal{M}\subseteq\{1,\dots,N\}$. Define $\mathcal{H}_\mathcal{M}\subseteq L^2(\R^{3N})$ as the space of functions which are symmetric in all variables in $\mathcal{M}$, i.e.~for $\psi\in\mathcal{H}_\mathcal{M}$,
$$\psi(z_1\mydots z_j\mydots z_k\mydots z_N)=\psi(z_1\mydots z_k\mydots z_j\mydots z_N) \qquad\forall\,j,k\in\mathcal{M}.$$
\end{definition}

\begin{lem}\label{lem:fqq}
Let $f:\mathbb{N}_0\rightarrow\mathbb{R}_0^+$ and $\mathcal{M}_1,\mathcal{M}_{1,2}\subseteq\{1,2 \mydots N\}$ with $1\in\mathcal{M}_1$ and $1,2\in\mathcal{M}_{1,2}$. 
\lemit{
	\item 	\label{lem:fqq:1}
			$\hat{n}^2=\frac{1}{N}\sum\limits_{j=1}^N q_j,$
	\item	\label{lem:fqq:2}
			$\norm{\hat{f}q_1\psi}^2\leq\frac{N}{|\mathcal{M}_1|}\norm{\hat{f}\hat{n}\psi}^2\qquad \forall \psi\in\mathcal{H}_{\mathcal{M}_1},$
	\item	\label{lem:fqq:3}
			$\norm{\hat{f}q_1q_2\psi}^2\leq\frac{N^2}{|\mathcal{M}_{1,2}|\left(|\mathcal{M}_{1,2}|-1\right)}\norm{\hat{f}\,\hat{n}^2\psi}^2 
			\qquad\forall \psi\in\mathcal{H}_{\mathcal{M}_{1,2}}.$
}
\end{lem}

\begin{proof}
Part (a) follows immediately from Corollary \ref{cor:p:1}. Consequently, for $\psi\in\mathcal{H}_{\mathcal{M}_1}$,
\begin{equation*}\begin{split}
\norm{\hat{f}\,\hat{n}\psi}^2&=\frac{1}{N}\sum\limits_{j=1}^N\llr{\psi,\hat{f}^2 q_j\psi}
\geq \frac{1}{N}\sum\limits_{j\in\mathcal{M}_1}\llr{\psi,\hat{f}^2 q_j\psi}
=\frac{|\mathcal{M}_1|}{N}\norm{\hat{f}q_1\psi}^2
\end{split}\end{equation*}
and analogously for $\psi\in\mathcal{H}_{\mathcal{M}_{1,2}}$,
$$\norm{\hat{f}\hat{n}^2\psi}\geq \frac{1}{N^2}\sum\limits_{j,k\in\mathcal{M}_{1,2}}\llr{\psi,\hat{f}^2q_jq_k\psi}
\geq \frac{|\mathcal{M}_{1,2}|(|\mathcal{M}_{1,2}|-1)}{N^2}\norm{\hat{f}q_1q_2\psi}^2.$$
\end{proof}
%Combining Lemma \ref{lem:fqq} with Lemma \ref{lem:commutators}, we arrive at the following estimates: 

\begin{cor} \label{cor:fqq}
Let $f:\mathbb{N}_0\rightarrow\mathbb{R}_0^+$ and $\mathcal{H}_{\mathcal{M}_1}$, $\mathcal{H}_{\mathcal{M}_{1,2}}$ as in Lemma~\ref{lem:fqq}. 
\lemit{
	\item	\label{cor:fqq:1}For $\psi\in\mathcal{H}_{\mathcal{M}_1}$,
			$$\norm{\nabla_1\hat{f}q_1\psi}\ls\onorm{\hat{f}}\norm{\nabla_1q_1\psi}\quad\text{and}\quad
			\norm{\partial_{x_1}\hat{f}\qp_1\psi}\ls\onorm{\hat{f}}\norm{\partial_{x_1}\qp_1\psi}.$$
	\item	\label{cor:fqq:2}For $\psi\in\mathcal{H}_{\mathcal{M}_{1,2}}$,
			$$\norm{\nabla_1\hat{f}q_1q_2\psi}\ls\onorm{\hat{f}\,\hat{n}}\norm{\nabla_1q_1\psi}\quad\text{and}\quad
			\norm{\partial_{x_1}\hat{f}\qp_1\qp_2\psi}\ls\onorm{\hat{f}\,\hat{n}}\norm{\partial_{x_1}\qp_1\psi}.$$
}
\end{cor}

\begin{proof}
Insertion of $1=p_1+q_1$ in front of $\nabla_1$ yields with Lemma~\ref{lem:commutators:2}
\begin{equation*}
\norm{\nabla_1\hat{f}q_1\psi}\leq(\onorm{\hat{f}}+\onorm{\hat{f}_1})\norm{\nabla_1q_1\psi}
\overset{\text{\ref{lem:l}}}{\ls}\onorm{\hat{f}}\norm{\nabla_1q_1\psi}
\end{equation*}
and
\begin{equation*}
\norm{\nabla_1\hat{f}q_1q_2\psi}\leq \norm{\hat{f}_1q_2\nabla_1q_1\psi}+\norm{\hat{f}q_2\nabla_1q_1\psi}
\ls \left(\onorm{\hat{f}_1\hat{n}}+\onorm{\hat{f}\hat{n}}\right)\norm{\nabla_1q_1\psi}
\end{equation*}
by Lemma~\ref{lem:fqq:2} as $\nabla_1q_1\psi\in\mathcal{H}_{\{2,\dots,N\}}$. As $n(k)\leq n(k+1)$,
$\onorm{\hat{f}_1\hat{n}}\leq \onorm{\hat{fn}_1}=\onorm{\hat{f}\,\hat{n}}$ by Lemma~\ref{lem:l:1}.
The respective second identities are shown analogously, using that $\qp q=\qp$ and that $\partial_{x_1}\qp_1\psi\in\mathcal{H}_{\{2,\dots,N\}}$.
\end{proof}

The next lemma provides an estimate of the difference between expectation values with respect to a symmetric $N$-body wavefunction $\psi$ and with respect to $\Phi(t)$.

\begin{lem}\label{lem:psi-Phi}
Let $\psi\in L_+^2(\R^{3N})$ be normalised and $f\in L^\infty(\R)$. Then
\begin{equation*}
\left|\llr{\psi,f(x_1)\psi}-\lr{\Phi(t),f\Phi(t)}_{L^2(\R)}\right|\ls\norm{f}_{L^\infty(\R)}\llr{\psi,\hat{n}\psi}.
\end{equation*}
\end{lem}

\begin{proof}
We drop the time dependence of $\Phi$. Inserting $1=p_1+q_1$ on both sides of $f(x_1)$ yields 
\begin{equation*}\begin{split}
\left|\llr{\psi,f(x_1)\psi}-\lr{\Phi,f\Phi}_{L^2(\R)}\right|\leq&
\left|\llr{\psi,p_1f(x_1)p_1\psi}-\lr{\Phi,f\Phi}_{L^2(\R)}\right|\\
&+|\llr{q_1\psi,f(x_1)q_1\psi}|+2|\llr{\psi,p_1f(x_1)q_1\psi}|.
\end{split}\end{equation*}
We estimate the first term as
\begin{equation*}\begin{split}
\left|\llr{\psi,\pc_1\ket{\Phi(x_1)}\bra{\Phi(x_1)}f(x_1)\ket{\Phi(x_1)}\bra{\Phi(x_1)}\pc_1\psi}-\lr{\Phi,f\Phi}_{L^2(\R)}\right|
&\leq |\lr{\Phi,f\Phi}_{L^2(\R)}\llr{\psi,q_1\psi}|\\
&\leq \norm{f}_{L^\infty(\R)}\llr{\psi,\hat{n}\psi}
\end{split}\end{equation*}
by Lemma~\ref{lem:fqq:1} and as $\hat{n}^2\leq\hat{n}$. The second term is bounded by
\begin{equation*}
|\llr{q_1\psi,f(x_1)q_1\psi}|\leq \norm{f}_{L^\infty(\R)} \norm{q_1\psi}^2\leq \norm{f}_{L^\infty(\R)}\llr{\psi,\hat{n}\psi}.
\end{equation*}
For the third term, we compute
\begin{equation*}\begin{split}
\left|\llr{\psi,p_1f(x_1)\hat{n}^\frac{1}{2}q_1\hat{n}^{-\frac{1}{2}}\psi}\right|&\overset{\text{\ref{lem:commutators:2}}}{=}
\left|\llr{\hat{n}_1^\frac{1}{2}p_1\psi,f(x_1)\hat{n}^{-\frac{1}{2}}q_1\psi}\right|\\
&\overset{\hphantom{\ref{lem:fqq:2}}}{\leq}\norm{f}_{L^\infty(\R)}\norm{\hat{n}_1^\frac{1}{2}\psi}\norm{\hat{n}^{-\frac{1}{2}}q_1\psi}\overset{\text{\ref{lem:fqq:2}}}{\ls}\norm{f}_{L^\infty(\R)}\llr{\psi,\hat{n}\psi},
\end{split}\end{equation*}
where we have used that $\sqrt{k+1}\leq \sqrt{k}+1$, hence $n_1(k)\leq n(k)+N^{-\frac{1}{2}}\leq 2n(k)\ls n(k)$. 
\end{proof}

In the following lemma, we estimate two particular scalar products.

\begin{lem}\label{lem:GammaLambda:and:O_12}
Let $O_{j,k}$ be an operator that acts nontrivially only on the $j$\textsuperscript{th} and $k$\textsuperscript{th} coordinate and let $F:\R^3\times\R^3\rightarrow\R^d$ for $d\in\mathbb{N}$. 
\lemit{
\item	\label{lem:GammaLambda}
		Let $\Gamma,\Lambda\in\mathcal{H}_\mathcal{M}$ for some $\mathcal{M}$ such that $j\notin\mathcal{M}$ and $k,l\in\mathcal{M}$. Then
		\begin{equation*}
		|\llr{\Gamma,O_{j,k}\Lambda}|\leq\norm{\Gamma}
		\Big(|\llr{O_{j,k}\Lambda,O_{j,l}\Lambda}|+|\mathcal{M}|^{-1}\norm{O_{j,k}\Lambda}^2\Big)^\frac12.
		\end{equation*}	
\item	\label{lem:O_12}
		Let $r_k$, $s_k$ and $t_j$ denote operators acting only on the factors $j$ and $k$ of the tensor product, respectively. Then for $j\neq k\neq l\neq j$,
		\begin{equation*}
		|\llr{r_k F(z_j,z_k)s_k t_j\Gamma,r_lF(z_j,z_l) s_l t_j \Gamma }|\leq\norm{s_kF(z_j,z_k)r_k t_j\Gamma}^2.
		\end{equation*}
}
\end{lem}

\begin{proof}
Using the symmetry of $\Gamma,\Lambda$ in all coordinates contained in $\mathcal{M}$, we find
\begin{equation*}\begin{split}
|\llr{\Gamma,O_{j,k}\Lambda}|&\leq \norm{\Gamma}\frac{1}{|\mathcal{M}|}\bigg\lVert\sum\limits_{m\in\mathcal{M}}O_{j,m}\Lambda\bigg\rVert\\
&\leq \norm{\Gamma}\left(\frac{1}{|\mathcal{M}|^2}\Bigg(\sum\limits_{\substack{n,m\in\mathcal{M}\\n\neq m}}\llr{O_{j,m}\Lambda,O_{j,n}\Lambda} + \sum\limits_{m\in\mathcal{M}}\norm{O_{j,m}\Lambda}^2\Bigg)\right)^\frac12\\
&=\norm{\Gamma}^2\left(\frac{|\mathcal{M}|-1}{|\mathcal{M}|}\llr{O_{j,k}\Lambda,O_{j,l}\Lambda}+\frac{1}{|\mathcal{M}|}\norm{O_{j,k}\Lambda}^2\right)^\frac{1}{2}.
\end{split}\end{equation*}
For part (b), we use that, for instance, $r_l$ and $F(z_j,z_k)$ commute, hence
\begin{align*}
|\llr{t_j\Gamma, s_kF(z_j,z_k)r_kr_lF(z_j,z_l)s_l t_j \Gamma}|
&=|\llr{r_l t_j\Gamma, s_k F(z_j,z_l) F(z_j,z_k) s_l r_k t_j\Gamma}|\\
&=|\llr{r_l t_j\Gamma, F(z_j,z_l) s_l s_k F(z_j,z_k) r_k t_j\Gamma}|\\
&\leq\norm{s_kF(z_j,z_k)r_k t_j\Gamma}^2.
\end{align*}
\end{proof}

The next lemma collects estimates for the time evolved condensate wavefunction.

\begin{lem}\label{lem:Phi}
$H^2(\R)$ solutions of the NLS equation \eqref{NLS} exist globally. 
\lemit{
	\item 	For any fixed time $t\in\R$, 
			\begin{align*}\begin{array}{ll}
			\norm{\Phi(t)}_{L^2(\R)}=1, \qquad &	\norm{\Phi(t)}_{H^1(\R)}\leq\mathfrak{e}(t),\\
			\norm{\Phi(t)}_{L^\infty(\R)}\leq\mathfrak{e}(t), \qquad & \norm{\Phi(t)}_{H^2(\R)}\ls \exp\left\{\mathfrak{e}^2(t)+\int_0^t\mathfrak{e}^2(s)\d s\right\}.\end{array}
			\end{align*}	
	\item	For sufficiently small $\varepsilon$ and fixed $t\in\R$,
			\label{cor:varphi:1} 
			\begin{align*}\begin{array}{lll}
			\norm{\chie}_{L^\infty(\R^2)}\ls \varepsilon^{-1}, \qquad & \norm{\nabla\chie}_{L^\infty(\R^2)}\ls\varepsilon^{-2},&\vphantom{\bigg(}\\
			\norm{\phe(t)}_{L^\infty(\R^3)}\ls\mathfrak{e}(t)\varepsilon^{-1}, \qquad &
			\norm{\na\phe(t)}_{L^\infty(\R^3)}\ls\mathfrak{e}(t)\varepsilon^{-2},\qquad &
			\norm{\nabla|\phe(t)|^2}_{L^2(\R^3)}\ls \mathfrak{e}(t)\varepsilon^{-2}.
			\end{array}
			\end{align*}
}
\end{lem}

\begin{proof}
For $\frac12<r\leq 4$ and $\Phi_0\in H^r(\R)$, \eqref{NLS} has a unique strong $H^r(\R)$-solution $\Phi\in\mathcal{C}(\R; H^r(\R))$ depending continuously on the initial data. The proof of this is sketched in Appendix \ref{appendix}. By assumption A4, $\Phi_0\in H^2(\R)$ and consequently $\Phi(t)\in H^2(\R)$. This implies $\frac{\d}{\d t}\norm{\Phi(t)}^2_{L^2(\R)}=0$ and by definition of $\mathcal{E}^{\Phi(t)}$ and $\mathfrak{e}(t)$,
\begin{align}
\norm{\Phi(t)}^2_{H^1(\R)}&\leq \mathcal{E}^{\Phi(t)}(t)+\norm{\Vp(t,\cdot)}_{L^\infty(\R^3)}\leq \mathfrak{e}^2(t).\label{H^1-bound}
\end{align}
Besides, $\Phi(t)\in H^2(\R)\subset\mathcal{C}^1(\R)$, hence
\begin{align*}
|\Phi(t,x)|^2
&=\int\limits_{-\infty}^x\left(\overline{\Phi'(t,\zeta})\Phi(t,\zeta)+\overline{\Phi(t,\zeta)}\Phi'(t,\zeta)\right)\d\zeta\\
&\leq \int\limits_{-\infty}^x\left(|\Phi'(t,\zeta)|^2+|\Phi(t,\zeta)|^2\right)\d\zeta=\norm{\Phi(t)}^2_{H^1(\R)}\leq\mathfrak{e}^2(t),\\
\norm{\tfrac{\partial}{\partial x}|\Phi(t)|^2}^2_{L^2(\R)}&\leq 4\int\limits_{\R}|\Phi'(t,x)|^2|\Phi(t,x)|^2\d x\leq 4\norm{\Phi(t)}_{L^\infty(\R)}^2\norm{\Phi(t)}^2_{H^1(\R)}\ls\mathfrak{e}^4(t).
\end{align*}
For $\Phi(t)\in H^4(\R)$, we obtain
\begin{align*}
\tfrac{\d}{\d t}\left(1+\norm{\dot{\Phi}(t)}^2_{L^2(\R)}\right)&=-2\Im\lr{\dot{\Vp}(t,(\cdot,0))\Phi(t),\dot{\Phi}(t)}_{L^2(\R)}-2\bb\Im\lr{\Phi(t)^2,\dot{\Phi}(t)^2}_{L^2(\R)}\\
&\leq 2\norm{\dot{\Vp}(t,\cdot)}_{L^\infty(\R^3)}(1+\norm{\dot{\Phi}(t)}^2_{L^2(\R)})+2\bb\norm{\Phi(t)}^2_{L^\infty(\R)}\norm{\dot{\Phi}(t)}^2_{L^2(\R)},
\end{align*}
hence by Grönwall's inequality and as $\norm{\Phi(t)}_{L^\infty(\R)}\leq\mathfrak{e}(t)$,
\begin{align*}
\norm{\dot{\Phi}(t)}^2_{L^2(\R)}
&\leq \left(1+\norm{\dot{\Phi}(0)}^2_{L^2(\R)}\right)\exp\left\{2\int_0^t\left(\norm{\dot{\Vp}(s,\cdot)}_{L^\infty(\R^3)}+\bb\mathfrak{e}^2(s)\right)\d s\right\}\\
&\ls\exp\left\{2\mathfrak{e}^2(t)+2\int_0^t\mathfrak{e}^2(s)\d s\right\}.
\end{align*}
This implies a bound for $\norm{\Phi(t)}_{H^2(\R)}$ because
\begin{align*}
\norm{\dot{\Phi}(t)}_{L^2(\R)}\geq \norm{\Phi''(t)}_{L^2(\R)}-\bb\norm{\Phi(t)}^2_{L^\infty(\R)}-\norm{\Vp(t,\cdot)}_{L^\infty(\R^3)}\gs\norm{\Phi''(t)}_{L^2(\R)}-\mathfrak{e}^2(t)
\end{align*}
and consequently
\begin{align*}
\norm{\Phi(t)}_{H^2(\R)}&\leq \norm{\Phi''(t)}_{L^2(\R)}+2\norm{\Phi(t)}_{H^1(\R)}\\
&\ls \mathfrak{e}^2(t)+\exp\left\{\mathfrak{e}^2(t)+\int_0^t\mathfrak{e}^2(s)\d s\right\}
\ls\exp\left\{\mathfrak{e}^2(t)+\int_0^t\mathfrak{e}^2(s)\d s\right\}.
\end{align*}
By continuity of the solution map, this bound extends to $\Phi(t)\in H^2(\R)$.
If the solution $\Phi(t)\in H^3(\R)\subset\mathcal{C}^2(\R)$, we find further
$$|\Phi'(x)|^2=
\int\limits_{-\infty}^x\left(\overline{\Phi'(\zeta})\Phi''(\zeta)+\overline{\Phi''(\zeta)}\Phi'(\zeta)\right)\d\zeta
\leq \norm{\Phi}_{H^2(\R)}^2,$$
which extends to $\Phi(t)\in H^2(\R)$ by continuity of the solution map.
For part (b), recall that $\chie(y)=\tfrac{1}{\varepsilon}\chi(\tfrac{y}{\varepsilon})$, hence $\norm{\chie}_{L^\infty(\R^2)}=\tfrac{1}{\varepsilon}\norm{\chi}_{L^\infty(\R^2)}\ls\tfrac{1}{\varepsilon}$ and analogously $\norm{\nabla\chie}_{L^\infty(\R^2)}\ls\tfrac{1}{\varepsilon^2}$. Together with (a), this implies the bounds for $\norm{\phe(t)}_{L^\infty(\R^3)}$ and $\norm{\nabla\phe(t)}_{L^\infty(\R^3)}$ as
\begin{align*}
|\na\phe(t,z)|^2
\leq|\Phi'(t,x)|^2|\chie(y)|^2+|\Phi(t,x)|^2|\nabla\chie(y)|^2
\ls\norm{\Phi(t)}^2_{H^2(\R)}\varepsilon^{-2}+\mathfrak{e}^2(t)\varepsilon^{-4}\ls\mathfrak{e}^2(t)\varepsilon^{-4}
\end{align*}
for any fixed time $t$ and $\varepsilon$ small enough. Finally,
\begin{align*}
\norm{\nabla|\phe(t)|^2}^2_{L^2(\R^3)}
=&\norm{\tfrac{\partial}{\partial x}|\Phi(t)|^2}^2_{L^2(\R)}\int\limits_{\R^2}|\chie(y)|^4\d y+\int\limits_{\R}|\Phi(t,x)|^4\d x\int\limits_{\R^2}|\nabla_y|\chie(y)|^2|^2\d y\\
\ls&\mathfrak{e}^4(t)\varepsilon^{-2}+4\mathfrak{e}^2(t)\int\limits_{\R^2}|\nabla_y\chie(y)|^2|\chie(y)|^2\d y
\ls\mathfrak{e}^2(t)\varepsilon^{-4}.
\end{align*}
\end{proof}

Now we prove some elementary facts enabling us to estimate one- and two-body potentials.

\begin{lem}\label{lem:pfp}
Let $t\in\R$ be fixed and let $j,k\in\{1\mydots N\}$. Let $g:\R^3\times\R^3\rightarrow\R$ be a measurable function such that $|g(z_j,z_k)|\leq G(z_k-z_j)$ almost everywhere for some $G:\R^3\rightarrow\R $. 

\lemit{
	\item	\label{lem:pfp:1}
			For $G\in L^1(\R^3)$,
			$$
			\onorm{p_jg(z_j,z_k)p_j}\ls\mathfrak{e}^2(t)\varepsilon^{-2}\norm{G}_{L^1(\R^3)}.
			$$
	\item	\label{lem:pfp:2}
			For $G\in L^2(\R^3)\cap L^\infty(\R^3)$,
%% L^infty damit der Multiplikationsoperator von G auch beschränkt ist
			$$\onorm{g(z_j,z_k)p_j}=\onorm{p_jg(z_j,z_k)}\ls \mathfrak{e}(t)\varepsilon^{-1}\norm{G}_{L^2(\R^3)}.$$
	\item	\label{lem:pfp:3}
			For $G\in L^2(\R^3)\cap L^\infty(\R^3)$,
			\begin{equation*}\begin{split}
			\onorm{g(z_j,z_k)\na_jp_j}&=\onorm{\ket{\phe(t,z_j)}\bra{\na\phe(t,z_j)}g(z_j,z_k)}
			\ls\mathfrak{e}(t)\varepsilon^{-2}\norm{G}_{L^2(\R^3)}.
			\end{split}\end{equation*}
	\item	\label{lem:pfp:4}
			Now let $g:\R\times\R\rightarrow\R$ be a measurable function such that $|g(x_j,x_k)|\leq G(x_k-x_j)$ almost everywhere for some 
			$G\in L^2(\R)\cap L^\infty(\R)$. Then
			\begin{align*}
			%\onorm{G(x_j)\pp_j}&=\onorm{\pp_jG(x_j)}\leq\mathfrak{e}(t)\norm{G}_{L^2(\R)},\\
			\onorm{g(x_j,x_k)\pp_j}&=\onorm{\pp_jg(x_j,x_k)}\leq\mathfrak{e}(t)\norm{G}_{L^2(\R)},\\
			\onorm{g(x_j,x_k)\partial_{x_j}\pp_j}&=\onorm{\ket{\Phi(t,x_j)}\bra{\partial_{x_j}\Phi(t,x_j)}g(x_j,x_k)}
			\leq \norm{\Phi}_{H^2(\R)}\norm{G}_{L^2(\R)}.
			\end{align*}
}

\end{lem}
\begin{proof}
Let $\psi\in L^2(\R^{3N})$ and drop the time dependence of $\phe$ and $\Phi$ in the notation. Then
\begin{align*}
\norm{p_j g(z_j,z_k)p_j\psi}&=\norm{\ket{\phe(z_j)}\bra{\phe(z_j)}g(z_j,z_k)\ket{\phe(z_j)}\bra{\phe(z_j)}\psi}\\
&\leq \int_{\R^3}|\phe(z_j)|^2|g(z_j,z_k)|\d z_j \;\norm{p_j\psi}\\
&\leq \norm{\phe}^2_{L^\infty(\R^3)}\int_{\R^3}|G(z_j-z_k)|\d z_j\;\norm{\psi}.
\end{align*}
The multiplication operators corresponding to $G$ and $g$ as well as $p_j$, $\na_jp_j$ and $\partial_{x_j}\pp_j$ are bounded. This implies the first equalities in (b) to (d).
The second equalities follow from
\begin{align*}
\onorm{g(z_j,z_k)p_j}^2&=\sup\limits_{\substack{\psi\in L^2(\R^{3N})\\\norm{\psi}=1}}\llr{\psi,p_j |g(z_j,z_k)|^2 p_j\psi}\leq \onorm{p_j|g(z_j,z_k)|^2p_j}\overset{\text{(a)}}{\ls} \norm{G}^2_{L^2(\R^3)}\mathfrak{e}^2(t)\varepsilon^{-2},\\
\onorm{G(x_j)\pp_j}^2&\leq \onorm{\pp_j|G(x_j)|^2\pp_j}\leq \norm{G}^2_{L^2(\R)}\norm{\Phi}_{L^\infty(\R)}^2,\\
\onorm{g(z_j,z_k)\na_jp_j}^2 &=
\sup\limits_{\substack{\psi\in L^2(\R^{3N})\\\norm{\psi}=1}}
\llr{\psi,\ket{\phe(z_j)}\bra{\na_j\phe(z_j)}g(z_j,z_k)|^2\ket{\na_j\phe(z_j)}\bra{\phe(z_j)}\psi}\\
&\leq \int\limits_{\R^3}|\na\phe(z_j)|^2G(z_k-z_j)^2\d z_j\;\onorm{p_j}^2
\leq \norm{\na\phe}^2_{L^\infty(\R^3)}\norm{G}^2_{L^2(\R^3)}
\end{align*}
and analogously for the second part of (d).
\end{proof}

\subsection{A priori estimate of the kinetic energy}
In this section, we prove estimates for the kinetic energy $\norm{\nabla_j\psi^{N,\varepsilon}(t)}$ and related quantities, which follow from the fact that the renormalised energy per particle $E^{\psi^{N,\varepsilon}(t)}(t)$ is bounded by $\mathfrak{e}(t)$. Particularly meaningful is assertion (a) of the following lemma: it states that the part of the wavefunction with one particle excited in the confined directions is of order $\varepsilon$. The lemma provides a sufficient estimate for most of the terms in Proposition \ref{prop:alpha}. To bound \eqref{gamma_b_3:2}, we require a better estimate (see Section \ref{subsec:E_kin}).

\begin{lem}\label{lem:a_priori}
Let $\varepsilon$ be small enough and $t\in\R$ be fixed. Then
\lemit{
	\item	\label{lem:a_priori:1}
			$\norm{\qc_1\psi^{N,\varepsilon}(t)}\leq \mathfrak{e}(t)\varepsilon, 
			\qquad \norm{\hat{l}\qc_1\psi^{N,\varepsilon}(t)}\leq\mathfrak{e}(t) N^\xi\varepsilon,$
	\item	\label{lem:a_priori:2}
			$
			\onorm{\partial_{x_1}\pp_1}\leq \mathfrak{e}(t), \quad
			\onorm{\partial_{x_1}^2\pp_1}\leq \norm{\Phi(t)}_{H^2(\R)}, \quad
			\onorm{\nabla_{y_1}\pc_1}\ls \varepsilon^{-1}, \quad
			\onorm{\na_1p_1}\ls\varepsilon^{-1},
		$	
	\item	\label{lem:a_priori:3}
			$\norm{\partial_{x_1}\qp_1\psi^{N,\varepsilon}(t)}\ls \mathfrak{e}(t),\qquad
			\norm{\na_1\qc_1\psi^{N,\varepsilon}(t)}\ls\mathfrak{e}(t), 
			\qquad \norm{\na_1 \hat{l}\qc_1\psi^{N,\varepsilon}(t)}\ls N^\xi\mathfrak{e}(t),
			$
	\item	\label{lem:a_priori:4}
			$\norm{\partial_{x_1}\psi^{N,\varepsilon}(t)}\leq\mathfrak{e}(t),   \qquad \norm{\nabla_{y_1}\psi^{N,\varepsilon}(t)}\ls \varepsilon^{-1},
			\qquad \norm{\na_1\psi^{N,\varepsilon}(t)}\ls\varepsilon^{-1},$
	\item 	\label{lem:a_priori:5}
		$
			%\norm{\na_1\pc_1\qp_1\psi^{N,\varepsilon}(t)}\ls\varepsilon^{-1}, \qquad
			\norm{\na_1\hat{l}\pc_1\qp_1\qp_2\psi^{N,\varepsilon}(t)}\ls\varepsilon^{-1},\qquad 
			\norm{\na_1\pc_1\qp_1\qc_2\psi^{N,\varepsilon}(t)}\ls \mathfrak{e}(t). 
			$
}
\end{lem}

\begin{proof} 
Abbreviating $\psi^{N,\varepsilon}(t)\equiv \psi$, we compute
\begin{align*}
E^{\psi}(t)&=\tfrac{1}{N}\llr{\psi,H_\beta(t)\psi}-\tfrac{E_0}{\varepsilon^2}\\
&=\llr{\psi,\tfrac{1}{N}\bigg(\sum\limits_{j=1}^N\left(-\partial_{x_j}^2+\left(-\Delta_{y_j}+\tfrac{1}{\varepsilon^2}V^\perp(\tfrac{y_j}{\varepsilon})-\tfrac{E_0}{\varepsilon^2}\right)+\Vp(t,\z_j)\right)+\sum\limits_{i<j}\wb(\z_i-\z_j)\bigg)\psi}\\
&\geq \norm{\partial_{x_1}\psi}^2+\llr{\qc_1\psi,\left(-\Delta_{y_1}+\tfrac{1}{\varepsilon^2}V^\perp(\tfrac{y_1}{\varepsilon})-\tfrac{E_0}{\varepsilon^2}\right)\qc_1\psi} -\norm{\Vp(t)}_{L^\infty(\R^3)}
\end{align*}
since $\wb\in\mathcal{W}_{\beta,\eta}$ is non-negative and $\left(-\Delta_{y_1}+\tfrac{1}{\varepsilon^2}V^\perp(\tfrac{y_1}{\varepsilon})-\tfrac{E_0}{\varepsilon^2}\right)\chie(y_1)=0 $.
$\tfrac{E_0}{\varepsilon^2}$ is the smallest eigenvalue of $-\Delta_{y_1}+\tfrac{1}{\varepsilon^2}V^\perp(\tfrac{y_1}{\varepsilon})$ and as a consequence of the rescaling by $\varepsilon$, the spectral gap to the next eigenvalue is of order $\varepsilon^{-2}$. Hence
$$\llr{\qc_1\psi,\left(-\Delta_{y_1}+\tfrac{1}{\varepsilon^2}V^\perp(\tfrac{y_1}{\varepsilon})-\tfrac{E_0}{\varepsilon^2}\right)\qc_1\psi}\gs\tfrac{1}{\varepsilon^2}\llr{\psi,\qc_1\psi},$$
which implies
\begin{equation}
\norm{\partial_{x_1}\psi}^2+\tfrac{1}{\varepsilon^2}\norm{\qc_1\psi}^2
\leq \norm{\Vp(t)}_{L^\infty(\R^3)}+|E^{\psi}(t)|\leq\mathfrak{e}^2(t).\label{eqn:a_priori:2}
\end{equation}
Besides, by assumption A2, $\norm{(V^\perp-E_0)_-}_{L^\infty(\R^2)}\ls 1$, hence
\begin{align*}
\mathfrak{e}^2(t)\geq&\llr{\qc_1\psi,\left(-\Delta_{y_1}+\tfrac{1}{\varepsilon^2}V^\perp(\tfrac{y_1}{\varepsilon})-\tfrac{E_0}{\varepsilon^2}\right)\qc_1\psi}\\
=&\norm{\nabla_{y_1}\qc_1\psi}^2+\tfrac{1}{\varepsilon^2}\llr{\qc_1\psi,\left(V^\perp(\tfrac{y_1}{\varepsilon})-E_0\right)_+\qc_1\psi}-\tfrac{1}{\varepsilon^2}\llr{\qc_1\psi,\left(V^\perp(\tfrac{y_1}{\varepsilon})-E_0\right)_-\qc_1\psi}\\
\geq &\norm{\nabla_{y_1}\qc_1\psi}^2-\tfrac{1}{\varepsilon^2}\norm{(V^\perp-E_0)_-}_{L^\infty(\R^2)}\norm{\qc_1\psi}^2
\gs\norm{\nabla_{y_1}\qc_1\psi}^2-\mathfrak{e}^2(t)
\end{align*}
and consequently $\norm{\nabla_{y_1}\qc_1\psi}^2\ls \mathfrak{e}^2(t).$
The remaining inequalities of (a) to (d) follow by Lemma~\ref{lem:l:2}, Lemma \ref{lem:commutators:2}, by using that $q_1^{(\Phi)}=1-p_1^{(\Phi)}$ and from
$\onorm{\partial_{x_1}p_1}\leq\onorm{\partial_{x_1}\pp_1}\leq\norm{\Phi'(t)}_{L^2(\R)}$ and $\onorm{\nabla_{y_1}\pc_1} \leq\norm{\nabla\chie}_{L^2(\R^2)}  \ls\varepsilon^{-1}$. 
For the second part of (d), note that
$$\norm{\nabla_{y_1}\psi}\leq\norm{\nabla_{y_1}\qc_1\psi}+\norm{\nabla_{y_1}\pc_1\psi}\ls\mathfrak{e}(t)+\varepsilon^{-1}\ls \varepsilon^{-1}$$
for sufficiently small $\varepsilon$ and fixed $t\in\R$.
Assertion (e) is a consequence of parts (a) to (d) and Corollary \ref{cor:fqq}, Lemma \ref{lem:l} and Lemma \ref{lem:fqq}:
\begin{align*}
%\norm{\na_1\pc_1\qp_1\psi}^2&\leq\norm{\partial_{x_1}\qp_1\psi}^2+\norm{\nabla_{y_1}\pc_1\psi}^2\ls\varepsilon^{-2},\\
\norm{\na_1\hat{l}\pc_1\qp_1\qp_2\psi}^2&\leq\norm{\partial_{x_1}\hat{l}\qp_1\qp_2\psi}^2+\onorm{\nabla_{y_1}\pc_1}^2\norm{\hat{l}\qp_1\qp_2\psi}^2\ls \mathfrak{e}^2(t)+\varepsilon^{-2}\norm{\hat{n}\psi}^2,\\
\norm{\na_1\pc_1\qp_1\qc_2\psi}^2&\leq \norm{\partial_{x_1}\qp_1\psi}^2+\onorm{\nabla_{y_1}\pc_1}^2\norm{\qc_2\psi}^2\ls\mathfrak{e}^2(t).
\end{align*}
\end{proof}

For the last lemma in this section, we make use of Lemma \ref{lem:a_priori:1} to prove an estimate which is crucial for the control of $\gamma_a(t)$.
\begin{lem}\label{lem:taylor}
Let $f:\R\times\R^3\rightarrow\R$ such that $f(t)\in\mathcal{C}^1(\R^3)$ and $\nabla_y f(t)\in L^\infty(\R^3)$ for any $t\in\R$. Then 
\lemit{
	\item $\norm{(f(t,z_1)-f(t,(x_1,0))\pc_1\psi^{N,\varepsilon}(t)}\leq\varepsilon\norm{\nabla_yf(t)}_{L^\infty(\R^3)},$
	\item $\norm{(f(t,z_1)-f(t,(x_1,0))\psi^{N,\varepsilon}(t)}\leq\varepsilon\left(\mathfrak{e}(t)\norm{f(t)}_{L^\infty(\R^3)}+\norm{\nabla_yf(t)}_{L^\infty(\R^3)}\right).$
}
\end{lem}

\begin{proof}
For the first part, we expand $f(t,(x_1,\cdot))$ around $y=0$, which yields
\begin{align*}
\norm{(f(t,z_1)-f(t,(x_1,0))\pc_1\psi^{N,\varepsilon}(t)}^2
&=\norm{\pc_1\psi^{N,\varepsilon}(t)}^2\int\limits_{\R^2}\d y_1|\chie(y_1)|^2\left(f(t,z_1)-f(t,(x_1,0))\right)^2\\
&\leq \tfrac{1}{\varepsilon^2}\int\limits_{\R^2}\d y_1|\chi(\tfrac{y_1}{\varepsilon})|^2\left(\int\limits_0^1\d s\nabla_yf(x_1,sy_1)\cdot y_1\right)^2\\
&\leq\varepsilon^2\int\limits_{\R^2}\d y|y|^2|\chi(y)|^2\norm{\nabla_yf(t)}_{L^\infty(\R^3)}^2
\ls\varepsilon^2\norm{\nabla_yf(t)}_{L^\infty(\R^3)}^2.
\end{align*}
The last step follows because $\chi$ decays exponentially by \cite[Theorem 1]{griesemer2004}
since $E_0<\sigma_\mathrm{ess}(\Delta_y+V^\perp)$ (A2).
To prove the second part, we insert $1=\qc_1+\pc_1$ and apply Lemma \ref{lem:a_priori:1}.
\end{proof}

\subsection{Proof of Proposition \ref{prop:alpha}}\label{subsec:prop:alpha}
Let us from now on drop the time dependence of $\Phi$, $\phe$ and $\psi^{N,\varepsilon}$ in the notation and further abbreviate $\psi^{N,\varepsilon}\equiv\psi$. The time derivative of $\alpha_\xi(t)$ is bounded by
\begin{equation}\label{eqn:dt_alpha:1}
\left|\tfrac{\d}{\d t}\alpha_\xi(t)\right|\leq\left|\tfrac{\d}{\d t}\llr{\psi,\hat{m}\psi}\right| +\left|\tfrac{\d}{\d t}\big|E^{\psi}(t)-\mathcal{E}^{\Phi}(t)\big|\right|.
\end{equation}
For the second term in~\eqref{eqn:dt_alpha:1}, we compute first
\begin{equation}\label{eqn:dt_alpha:2}
\left|\tfrac{\d}{\d t}\big(E^{\psi}(t)-\mathcal{E}^{\Phi}(t)\big)\right|
=\left|\llr{\psi,\dot{\Vp}(t,\z_1)\psi}-\lr{\Phi,\dot{\Vp}\left(t,(x,0)\right)\Phi}_{L^2(\R)}\right|.
\end{equation}
By \cite[Theorem 6.17]{lieb_loss}, $\left|\tfrac{\d}{\d t}\big|E^{\psi}(t)-\mathcal{E}^{\Phi}(t)\big|\right| = \left|\tfrac{\d}{\d t}\big(E^{\psi}(t)-\mathcal{E}^{\Phi}(t)\big)\right|$ for almost every $t\in\R$ because $t\mapsto\tfrac{\d}{\d t}\big(E^{\psi}(t)-\mathcal{E}^{\Phi}(t)\big)$ is continuous due to assumption A3.
The first term in~\eqref{eqn:dt_alpha:1} yields
\begin{align}
\tfrac{\d}{\d t}\llr{\psi,\hat{m}\psi}
\overset{\text{\ref{lem:derivative_m:3}}}{=}&\;i\llr{\psi,\Big[H_\beta(t)-\sum\limits_{j=1}^N h_j(t),\hat{m}\Big]\psi}\nonumber\\
\overset{\text{\ref{lem:derivative_m:2}}}{=}
&\;\,iN\llr{\psi,\left[\Vp(t,\z_1)-\Vp\left(t,(x_1,0)\right),\hat{m}\right]\psi}+ i\tfrac{N(N-1)}{2}\llr{\psi,\left[Z_\beta^{(12)},\hat{m}\right]\psi}\nonumber\\
\overset{\text{\ref{lem:commutators:5}}}{=}
&\;\,iN\llr{\psi,\left[\Vp(t,\z_1)-\Vp\left(t,(x_1,0)\right),\hat{m}\right]\psi}&\label{eqn:dt_alpha:3}\\
&+ i\tfrac{N(N-1)}{2}\llr{\psi,\left[Z_\beta^{(12)},Q_0(\hat{m}-\hat{m}_2)+Q_1(\hat{m}-\hat{m}_1)\right]\psi},\label{eqn:dt_alpha:4}
\end{align}
where $Q_0:=p_1p_2$, $Q_1:=p_1q_2+q_1p_2$ and $Q_2:=q_1q_2$. To expand~\eqref{eqn:dt_alpha:4}, we write the commutator explicitly and insert $1=Q_0+Q_1+Q_2$ appropriately before or after $Z_\beta^{(12)}$. Terms with the same $Q_\mu$ on both sides cancel as a consequence of Lemma~\ref{lem:commutators:2}. Hence
\begin{align}
\frac{\eqref{eqn:dt_alpha:4}}{N(N-1)}=&
\tfrac{i}{2}\llr{\psi,\Big((Q_1+Q_2)Z_\beta^{(12)}(\hat{m}-\hat{m}_2)Q_0-Q_0(\hat{m}-\hat{m}_2)Z_\beta^{(12)}(Q_1+Q_2)\Big)\psi}\nonumber\\
&+\tfrac{i}{2}\llr{\psi,\Big((Q_0+Q_2)Z_\beta^{(12)}(\hat{m}-\hat{m}_1)Q_1-Q_1(\hat{m}-\hat{m}_1)Z_\beta^{(12)}(Q_0+Q_2)\Big)\psi}\nonumber\\
=&\tfrac{i}{2}\llr{\psi,\Big(Q_1(\hat{m}_{-1}-\hat{m}_1)Z_\beta^{(12)}Q_0+Q_2(\hat{m}_{-2}-\hat{m})Z_\beta^{(12)}Q_0\Big)\psi}
\nonumber\\
&-\tfrac{i}{2}\llr{\psi,\Big(Q_0Z_\beta^{(12)}(\hat{m}_{-1}-\hat{m}_1)Q_1+Q_0Z_\beta^{(12)}(\hat{m}_{-2}-\hat{m})Q_2\Big)\psi}
\nonumber\\
&+\tfrac{i}{2}\llr{\psi,\Big(Q_0Z_\beta^{(12)}(\hat{m}-\hat{m}_1)Q_1+Q_2(\hat{m}_{-1}-\hat{m})Z_\beta^{(12)}Q_1\Big)\psi}
\nonumber\\
&-\tfrac{i}{2}\llr{\psi,\Big(Q_1(\hat{m}-\hat{m}_1)Z_\beta^{(12)}Q_0+Q_1Z_\beta^{(12)}(\hat{m}_{-1}-\hat{m})Q_2\Big)\psi}
\nonumber\\
=&\Im\llr{\psi,Q_1(\hat{m}-\hat{m}_{-1})Z_\beta^{(12)}Q_0\psi}+\Im\llr{\psi,Q_2(\hat{m}-\hat{m}_{-2}) Z_\beta^{(12)}Q_0\psi}\nonumber\\
&+\Im\llr{\psi,Q_2(\hat{m}-\hat{m}_{-1}) Z_\beta^{(12)}Q_1\psi}.\nonumber
\end{align}
To simplify this expression, note that
\begin{align}
\hat{m}-\hat{m}_{-1}&=\sum\limits_{k=0}^Nm(k)P_k-\sum\limits_{k=1}^Nm(k-1)P_k=\sum\limits_{k=1}^N\left(m(k)-m(k-1)\right)P_k +m(0)P_0\nonumber\\
&=-\hat{m}^a_{-1}+m(0)P_0\nonumber
\end{align}
and analogously
\begin{equation*}
\hat{m}-\hat{m}_{-2}=-\hat{m}^b_{-2}+m(0)P_0+m(1)P_1.
\end{equation*}
Using that $Q_1P_0=Q_2P_0=Q_2P_1=0$, we obtain consequently
\begin{align}
\frac{\eqref{eqn:dt_alpha:4}}{N(N-1)}=&
-2\Im\llr{\psi,q_1p_2\hat{m}^a_{-1}Z_\beta^{(12)}p_1p_2\psi}\label{eqn:dt_alpha:5}\\
&-\Im\llr{\psi,q_1q_2\hat{m}^b_{-2}Z_\beta^{(12)}p_1p_2\psi}\label{eqn:dt_alpha:6}\\
&-2\Im\llr{\psi,q_1q_2\hat{m}^a_{-1}Z_\beta^{(12)}p_1q_2\psi}\label{eqn:dt_alpha:7},
\end{align}
where we have in~\eqref{eqn:dt_alpha:5} and~\eqref{eqn:dt_alpha:7} exploited the symmetry of $\psi$ in coordinates $1$ and $2$.
According to Corollary~\ref{cor:p:3}, $q=\qc+\qp\pc$, hence
\begin{align}
\eqref{eqn:dt_alpha:5}=&-2\Im\llr{\qc_1\psi,p_2\hat{m}^a_{-1}\wbot p_1p_2\psi}\label{eqn:dt_alpha:8}\\
&-2\Im\llr{\psi,\qp_1\hat{m}^a_{-1}\pc_1p_2Z_\beta^{(12)}p_1p_2\psi}\label{eqn:dt_alpha:9}.
\end{align}
In \eqref{eqn:dt_alpha:8}, we have used that the contribution of $|\Phi(x_1)|^2+|\Phi(x_2)|^2$ vanishes as $\qc_1|\Phi(x_1)|^2\pc_1=\qc_1|\Phi(x_2)|^2\pc_1=0$.
Similarly, we expand~\eqref{eqn:dt_alpha:6} and~\eqref{eqn:dt_alpha:7} into terms containing $\qc$ and terms containing $\pc_1\pc_2\wbot \pc_1\pc_2$:
\begin{align}
\eqref{eqn:dt_alpha:6}=&-\Im\llr{\qc_1\psi,q_2\hat{m}^b_{-2}\wbot p_1p_2\psi}
-\Im\llr{\qc_2\psi,\qp_1\pc_1\hat{m}^b_{-2}\wbot p_1p_2\psi}\nonumber\\
&-\Im\llr{\psi,\qp_1\qp_2\hat{m}^b_{-2}\pc_1\pc_2\wbot p_1p_2\psi}\nonumber\\
=&-\Im\llr{\qc_1\psi,q_2(1+\pc_2)\hat{m}^b_{-2}\wbot p_1p_2\psi}\label{eqn:dt_alpha:10}\\
&-\Im\llr{\psi,\qp_1\qp_2\hat{m}^b_{-2}\pc_1\pc_2\wbot p_1p_2\psi}\label{eqn:dt_alpha:11}
\end{align}
and
\begin{align}
\eqref{eqn:dt_alpha:7}=&-2\Im\llr{\qc_1\psi,q_2\hat{m}^a_{-1}\wbot p_1q_2\psi}\label{eqn:dt_alpha:12}\\
&-2\Im\llr{\qc_2\psi,\qp_1\pc_1\hat{m}^a_{-1}\wbot p_1q_2\psi}\label{eqn:dt_alpha:13}\\
&-2\Im\llr{\psi,\qp_1\qp_2\hat{m}^a_{-1}\pc_1\pc_2\wbot p_1\qc_2\psi}\label{eqn:dt_alpha:14}\\
&-2\Im\llr{\psi,\qp_1\qp_2\hat{m}^a_{-1}\pc_1\pc_2\wbot p_1\pc_2\qp_2\psi} \label{eqn:dt_alpha:15}\\
&+\tfrac{2\bb}{N-1}\Im\llr{\psi,q_1q_2\hat{m}^a_{-1}|\Phi(x_1)|^2p_1q_2\psi}.\label{eqn:dt_alpha:16}
\end{align}
Finally, we insert $1=p_1+q_1$ on both sides of the commutator in~\eqref{eqn:dt_alpha:3} and apply Lemma~\ref{lem:commutators:2}. Analogously to above, we obtain
\begin{align}
\eqref{eqn:dt_alpha:3}=&iN\llr{\psi,(p_1+q_1)\big(\Vp(t,\z_1)-\Vp(t,(x_1,0))\big)\hat{m}(p_1+q_1)\psi}\nonumber\\
&-iN\llr{\psi,(p_1+q_1)\hat{m}\big(\Vp(t,\z_1)-\Vp(t,(x_1,0))\big)(p_1+q_1)\psi}\nonumber\\
=&-2N\Im\llr{\psi,q_1\hat{m}^a_{-1}\big(\Vp(t,\z_1)-\Vp(t,(x_1,0))\big)p_1\psi}.\label{eqn:dt_alpha:19}
\end{align}
Collecting and regrouping all terms arising from~\eqref{eqn:dt_alpha:1} yields 
$\gamma_a=\eqref{eqn:dt_alpha:2}+\eqref{eqn:dt_alpha:19}$, 
$\gamma_b=\eqref{eqn:dt_alpha:4}$,
$\gamma_b^{(1)}=N(N-1)~\eqref{eqn:dt_alpha:9}$,
$\gamma_b^{(2)}=N(N-1)\big[\big(\eqref{eqn:dt_alpha:8}+\eqref{eqn:dt_alpha:10}\big)+\big(\eqref{eqn:dt_alpha:12}+\eqref{eqn:dt_alpha:13}\big)+\eqref{eqn:dt_alpha:14}\big]$
and $\gamma_b^{(3)}=N(N-1)\big(\eqref{eqn:dt_alpha:11}+\eqref{eqn:dt_alpha:15}+\eqref{eqn:dt_alpha:16}\big)$.
\qed

\subsection{Proof of Proposition \ref{prop:gamma}}\label{subsec:prop:gamma}

\subsubsection{Proof of the bound for $\gamma_a(t)$}\label{subsec:gamma_a}
As $2N\hat{m}^a_{-1}\ls\hat{l}$ for $\hat{l}$ from Lemma~\ref{lem:l}, we obtain with Lemma \ref{lem:psi-Phi}, Lemma \ref{lem:taylor}, Lemma \ref{lem:fqq:2} and Lemma \ref{lem:l:2} 
\begin{align*}
|\eqref{gamma_a:1}|
&\ls\left|\llr{\psi,\left(\dot{\Vp}(t,z_1)-\dot{\Vp}(t,(x_1,0))\right)\psi}\right|+\left|\llr{\psi,\dot{\Vp}(t,(x_1,0))\psi}-\lr{\Phi,\dot{\Vp}(t,(x,0))\Phi}_{L^2(\R)}\right|\\
&\ls\mathfrak{e}^3(t)\varepsilon+\mathfrak{e}(t)\llr{\psi,\hat{n}\psi},\\
|\eqref{gamma_a:2}|&\ls \norm{\hat{l}\hat{n}\psi}\norm{(\Vp(t,z_1)-\Vp(t,(x_1,0)))\pc_1\psi}\ls\mathfrak{e}^2(t)\varepsilon.
\end{align*}
\qed

\subsubsection{Proof of the bound for $\gamma_b^{(1)}(t)$}\label{subsec:gamma:b:1}
To estimate $\gamma_b^{(1)}$, we need to prove that $N\wb$ is close to the effective potential $\bb|\Phi|^2$.
As $(N-1)\hat{m}^a_{-1}\leq \hat{l}$, we obtain 
\begin{align*}
\big|\gamma_b^{(1)}\big|
&\ls \left|\llr{\hat{l}\qp_1\psi,\pc_1p_2\left(N\wbot -b_{N,\varepsilon}|\Phi(x_1)|^2+(b_{N,\varepsilon}-\tfrac{N}{N-1}\bb)|\Phi(x_1)|^2\right)p_1p_2\psi}\right|\\
&\overset{\text{\ref{lem:Phi}}}{\ls} \left|\llr{\hat{l}\qp_1\psi,\pc_1p_2\left(N\wbot -b_{N,\varepsilon}|\Phi(x_1)|^2\right)\pc_1p_2\pp_1\psi}\right|
+ \left((\tfrac{N}{\varepsilon^2})^{-\eta}+N^{-1}\right)\mathfrak{e}^2(t)
\end{align*}
for $\mu$ small enough and with $\eta$ from Definition \ref{def:W} since $\wb\in\mathcal{W}_{\beta,\eta}$ and as $\norm{\hat{l}\qp_1\psi}\ls 1$ by Lemma~\ref{lem:l:1}. 
Writing the action of the projectors explicitly, we obtain by definition of $b_{N,\varepsilon}$
\begin{align*}
\pc_1p_2b_{N,\varepsilon}|\Phi(x_1)|^2\pc_1p_2 & = b_{N,\varepsilon}|\Phi(x_1)|^2\pc_1p_2=
N\Bigg(\,\int\limits_{\R^2}\d{y}'_1|\chie(y'_1)|^4|\Phi(x_1)|^2\norm{\wb}_{L^1(\R^3)}\Bigg)\pc_1p_2,\\
\pc_1p_2N\wbot \pc_1p_2&=N \Bigg(\,\int\limits_{\R^2}\d y'_1|\chie(y'_1)|^2\int\limits_{\R^3}\d z'_2
|\phe(z'_2)|^2\wb(z_1''-z'_2)\Bigg)\pc_1p_2,
\end{align*}
where $z''_1:=(x_1,y'_1)$. The substitution $z'_2\mapsto z:=z''_1-z'_2$ and subtraction of both lines leads to
\begin{equation}\label{eqn:Gamma(x_1)}
\Gamma(x_1):=N\int\limits_{\R^2}|\chie(y'_1)|^2\d y'_1
\Bigg(\,\int\limits_{\R^3}|\phe(z''_1-z)|^2\wb(\z)\d z
-|\phe(z''_1)|^2\norm{\wb}_{L^1(\R)}\Bigg).
\end{equation}
Let us first consider an analogous expression where $|\phe|^2$ is replaced by some $g\in C^\infty_0(\R^3)$. Expanding $g(z_1''-\cdot)$ around $z_1''$ yields
\begin{align}
\int\limits_{\R^3} g(z''_1-z)\wb(\z)\d z
&=g(z''_1)\norm{\wb}_{L^1(\R^3)}-\int\limits_{\R^3} \d z\int\limits_0^1\nabla g(z''_1-sz)\cdot z\wb(\z)\d s\nonumber\\
&=:g(z''_1)\norm{\wb}_{L^1(\R^3)}+R(z''_1),\nonumber
\end{align}
where
$$|R(z_1'')|\leq\sup\limits_{\substack{s\in[0,1]\\z\in\R^3}}|\nabla g(z_1''-sz)|\int\limits_{\R^3}\d z|z|\wb(z).$$
Hence 
$$\norm{R}^2_{L^2(\R^3)}\ls\varepsilon^4 N^{-2}\mu^2\norm{\nabla g}^2_{L^2(\R^3)}$$
because $|z|\leq \mu$ for $\z\in\supp\wb$ and as $\wb\in\mathcal{W}_{\beta,\eta}$ implies 
\begin{equation}\label{eqn:int:wb}
\int_{\R^3}\wb(\z)\d\z\ls\varepsilon^2N^{-1} b_{N,\varepsilon}=\varepsilon^2N^{-1}(b_{N,\varepsilon}-\bb)+\varepsilon^2N^{-1}\bb\ls\varepsilon^2N^{-1}.
\end{equation}
Consequently,
\begin{align*}
&\Bigg\lVert N\int\limits_{\R^2}|\chie(y'_1)|^2\d y'_1 \Bigg(\,
\int\limits_{\R^3} g(z''_1-z)\wb(\z)\d z-g(z''_1)\norm{\wb}_{L^1(\R^3)}\bigg)\Bigg\rVert^2_{L^2(\R)}\\
&\qquad\leq N^2\int\limits_\R \d x_1\left|\,\int\limits_{\R^2}\d y'_1|\chie(y'_1)|^2R(z''_1)\right|^2
\leq N^2\norm{|\chie|^2}^2_{L^2(\R^2)}\norm{R}^2_{L^2(\R^3)}
\ls\mu^2\varepsilon^2\norm{\nabla g}^2_{L^2(\R^3)},
\end{align*}
where we have in the second step used Hölder's inequality. By density, this bound extends to $g\in H^1(\R^3)$ and in particular to $g\equiv |\phe|^2$, hence
\begin{equation}\label{eqn:norm:Gamma}
\norm{\Gamma}_{L^2(\R)}
\ls \mu\varepsilon\norm{\nabla|\phe|^2}_{L^2(\R^3)}\overset{\text{\ref{lem:Phi}}}{\ls}\tfrac{\mu}{\varepsilon}\mathfrak{e}(t)
\end{equation}
and
$$
\big|\gamma_b^{(1)}\big|\leq\norm{\hat{l}\qp_1\psi}\onorm{\pp_1\Gamma(x_1)}+\left(N^{-1}+(\tfrac{N}{\varepsilon^2})^{-\eta}\right)\mathfrak{e}^2(t)\overset{\text{\ref{lem:pfp:4}}}{\ls} \left(\tfrac{\mu}{\varepsilon}+(\tfrac{N}{\varepsilon^2})^{-\eta}+N^{-1}\right)\mathfrak{e}^2(t).
$$
\qed

\subsubsection{Proof of the bound for $\gamma_b^{(2)}(t)$}\label{subsec:gamma_b^2}
Let us first define the functions needed for the integration by parts of the interaction.
\begin{definition}\label{def:he}
Define $\he:\R^3\rightarrow\R$ by
\[
\he(\z):=\begin{cases}
 \frac{1}{4\pi}\left(\;\int\limits_{\R^3}\frac{\wb(\zeta)}{|\z-\zeta|}\d\zeta-\int\limits_{\R^3}\frac{\varepsilon}{|\zeta|}\frac{\wb(\zeta)}{|\zeta^*-\z|}\d\zeta\right)& \text{for }|\z| < \varepsilon,\\
	0 & \text{else}
\end{cases}
\]
where 
\begin{equation*}
\zeta^*:=\tfrac{\varepsilon^2}{|\zeta|^2}\zeta.
\end{equation*}
We will abbreviate
\begin{equation*}
\heij:=\he(\z_i-\z_j).
\end{equation*}
\end{definition}

%The following lemma shows that $\he$ is indeed the correct function for the integration by parts and establishes estimates for $\he$ and $\na\he$.

\begin{lem}\label{lem:he}
Let $\mu\ll\varepsilon$. Then
\lemit{
	\item	\label{lem:he:1}
			$\he$ solves the boundary value problem
			\begin{equation}\label{eqn:boundary_problem}
			\begin{cases}
				\Delta\he(\z)=\wb(\z) 	& \text{for } \z\in B_{\varepsilon}(0),\vphantom{\bigg(}\\
				\hphantom{\Delta}\he(\z)=0			& \text{for } \z\in\partial B_{\varepsilon}(0),
			\end{cases}
			\end{equation}
			where $B_\varepsilon(0):=\{\z\in\R^3:|\z|<\varepsilon\}$.
%	\item	\label{lem:he:2}
%			$$\norm{\he}_{L^\infty(\R^3)}\ls N^{-1}\mu^{-1}\varepsilon^2,\qquad 
%			\norm{\he}_{L^2(\R^3)}\ls N^{-1}\varepsilon^{\frac52}.$$
	\item	\label{lem:he:3}
			$\norm{\na\he}_{L^\infty(\R^3)}\ls N^{-1}\mu^{-2}\varepsilon^2, \qquad
			\norm{\na\he}_{L^2(\R^3)}\ls N^{-1}\mu^{-\frac12}\varepsilon^2.$
}
\end{lem}
\begin{proof}
Green's function for the problem \eqref{eqn:boundary_problem} is 
$G(z,\zeta)=\tfrac{1}{4\pi}\left(\frac{1}{|\zeta-z|}-\frac{\varepsilon}{|\zeta|}\frac{1}{|z-\zeta^*|}\right),$ hence $\he\big|_{\overline{B_{\varepsilon}(0)}}$ is the unique solution of \eqref{eqn:boundary_problem}.
For part (b), define 
$$h^{(1)}(\z):=\begin{cases} \int\limits_{\R^3}\frac{\wb(\zeta)}{|\z-\zeta|}\d\zeta & \text{ for }|z|<\varepsilon,\\
0 & \text{ else},
\end{cases} \quad\text{ and }\quad
h^{(2)}(\z):=\begin{cases} \int\limits_{\R^3}\frac{\varepsilon}{|\zeta|}\frac{\wb(\zeta)}{|\zeta^*-\z|}\d\zeta & \text{ for }|z|<\varepsilon,\\
0 & \text{ else},
\end{cases}	
$$
hence $\he(\z)=:\tfrac{1}{4\pi}\left(h^{(1)}(\z)+h^{(2)}(\z)\right).$ We estimate $h^{(1)}$ and $h^{(2)}$ separately.\\

\noindent\emph{Estimate of $|\na h^{(1)}|$.}
Let $|\z|\leq 2\mu$ and substitute $\zeta\mapsto\zeta':=\zeta-\z$. As $\supp\wb\subseteq \overline{B_\mu(0)}$, we conclude that $|\zeta'|\leq|\zeta|+|\z|\leq 3\mu$ for $\zeta\in\supp\wb$ and consequently
\begin{align*}
%|h^{(1)}(\z)|&\leq \norm{\wb}_{L^\infty(\R^3)}\int\limits_{|\zeta|\leq\mu}\frac{1}{|\z-\zeta|}\d\zeta
%\ls \left(\tfrac{N}{\varepsilon^2}\right)^{-1+3\beta}\int\limits_{|\zeta'|\leq 3\mu}\frac{1}{|\zeta'|}\d\zeta'
%\ls N^{-1}\varepsilon^2 \mu^{-1},\\
|\na h^{(1)}(\z)|&\leq\norm{\wb}_{L^\infty(\R^3)}\int\limits_{|\zeta|\leq\mu}\frac{1}{|\z-\zeta|^2}\d\zeta
\ls \left(\tfrac{N}{\varepsilon^2}\right)^{-1+3\beta}\int\limits_{|\zeta'|\leq 3\mu}\frac{1}{|\zeta'|^2}\d\zeta'
\ls N^{-1}\varepsilon^2\mu^{-2}.
\end{align*}
For $2\mu\leq |\z|< \varepsilon$, note that $\zeta\in\supp\wb$ implies $|\zeta|\leq \mu\leq \frac{1}{2}|\z|$, hence
$|\z-\zeta|\geq|\z|-|\zeta|\geq \tfrac12|\z|$
and consequently 
\begin{align*}
%h^{(1)}(\z)&\leq\frac{2}{|\z|}\int\limits_{\R^3}\wb(\zeta)\d\zeta
%\ls N^{-1}\varepsilon^2|\z|^{-1}\ls N^{-1}\varepsilon^2\mu^{-1},\\
|\na h^{(1)}(\z)|&\leq\frac{4}{|\z|^2}\int\limits_{\R^3}\wb(\zeta)\d\zeta\ls N^{-1}\varepsilon^2|\z|^{-2}
\ls N^{-1}\varepsilon^2\mu^{-2}
\end{align*}
due to \eqref{eqn:int:wb}.
Hence
\begin{align*}
%\int\limits_{\R^3}|h^{(1)}(z)|^2\d z
%&\ls \int\limits_{|\z|\leq 2\mu}N^{-2}\varepsilon^{4}\mu^{-2}\d\z
%+\int\limits_{2\mu\leq|\z|\leq \varepsilon}N^{-2}\varepsilon^4\frac{1}{|\z|^2}\d\z
%\ls N^{-2}\varepsilon^5,\\
\int\limits_{\R^3}\big|\na h^{(1)}(z)\big|^2\d z
&\ls \int\limits_{|\z|\leq 2\mu}N^{-2}\varepsilon^{4}\mu^{-4}\d\z
+\int\limits_{2\mu\leq|\z|\leq \varepsilon}N^{-2}\varepsilon^4\frac{1}{|\z|^4}\d\z
\ls N^{-2}\varepsilon^4\mu^{-1}.
\end{align*}
\emph{Estimate of $|\na h^{(2)}|$.}
$\zeta\in\supp\wb$ implies $|\zeta|\leq \mu$, hence $|\zeta^*|=\frac{\varepsilon^2}{|\zeta|}\geq\frac{\varepsilon^2}{\mu}$. For $\mu$ sufficiently small that $\frac{\varepsilon}{\mu}>2$, we observe 
$|\z|<\varepsilon<\frac{1}{2}\frac{\varepsilon^2}{\mu}\leq\frac{1}{2}|\zeta^*|$
and consequently $|\zeta^*-\z|\geq |\zeta^*|-|\z| > \frac{1}{2}|\zeta^*|=\frac{1}{2}\frac{\varepsilon^2}{|\zeta|}$. This yields
\begin{align*}
%\big|h^{(2)}(\z)\big|&=\int\limits_{\R^3}\frac{\varepsilon}{|\zeta|}\frac{\wb(\zeta)}{|\zeta^*-\z|}\d\zeta
%<\tfrac{2}{\varepsilon}\int\limits_{\R^3}\wb(\zeta)\d\zeta\ls N^{-1}\varepsilon<N^{-1}\varepsilon^2\mu^{-1},\\
\big|\na h^{(2)}(z)\big|&=\int\limits_{\R^3}\frac{\varepsilon}{|\zeta|}\frac{\wb(\zeta)}{|\zeta^*-\z|^2}\d\zeta
\ls\varepsilon^{-3}\norm{\wb}_{L^\infty(\R^3)}\int\limits_{|\zeta|\leq\mu}|\zeta|\d\zeta
\ls N^{-1}\varepsilon^{-1}{\mu}<N^{-1}\varepsilon^2\mu^{-2}
\end{align*}
and consequently
$
%\int\limits_{\R^3}|h^{(2)}(\z)|^2\d z\ls N^{-2}\varepsilon^5, \qquad
\int_{\R^3}\big|\na h^{(2)}(\z)\big|^2\d z\ls N^{-2}\mu^2\varepsilon<N^{-2}\varepsilon^4\mu^{-1}.
$
\end{proof}
Besides, we need a smooth step function to prevent contributions from the boundary when integrating by parts over the ball $\overline{B_\varepsilon(0)}$.

\begin{definition}\label{def:theta}
Define $\te:\R^3\rightarrow [0,1]$ by
\begin{equation*}
\te(\z)=\begin{cases}
	1 & \text{for }|\z|\leq\mu,\\
	\theta_\varepsilon(|\z|) & \text{for } \mu<|\z|<\varepsilon,\\
	0 &  \text{for }|\z|\geq \varepsilon,
\end{cases}
\end{equation*}
where $\theta_\varepsilon:[\mu,\varepsilon]\rightarrow [0,1]$ is given by
\begin{equation}\label{def:theta_smooth}
\theta_\varepsilon(x):=\frac{\exp\left(-\frac{\varepsilon-\mu}{\varepsilon-x}\right)}
{\exp\left(-\frac{\varepsilon-\mu}{\varepsilon-x}\right)+\exp\left(-\frac{\varepsilon-\mu}{x-\mu}\right)}.
\end{equation}
Clearly, $\theta_\varepsilon$ is a smooth, decreasing  function with $\theta_\varepsilon(\mu)=1$ and $\theta_\varepsilon(\varepsilon)=0$.
We will write
\begin{equation*}
\teij:=\te(\z_i-\z_j).
\end{equation*}
\end{definition}

\begin{lem}\label{lem:theta}
Let $\mu\ll\varepsilon$. Then 
\lemit{
	\item	\begin{tabular}{p{4cm}p{4cm}}
			$\norm{\te}_{L^\infty(\R^3)}= 1$,
			&$\norm{\te}_{L^2(\R^3)}\ls \varepsilon^{\frac32},$
			\end{tabular}
	\item	\begin{tabular}{p{4cm}p{4cm}}
			$\norm{\na\te}_{L^\infty(\R^3)}\ls \varepsilon^{-1}$,&
			$\norm{\na\te}_{L^2(\R^3)}\ls \varepsilon^{\frac12}$.
			\end{tabular}
}
\end{lem}

\begin{proof}
Part (a) follows immediately from the definition of $\te$. For part (b), observe that
$|\frac{\d}{\d x}\theta_\varepsilon(x)|\leq 2(\varepsilon-\mu)^{-1}=2\varepsilon^{-1}(1-\frac{\mu}{\varepsilon})\ls\varepsilon^{-1}$. 
\end{proof}

%The following corollary collects estimates of $\he$ and $\te$ needed for the proof.
 
\begin{cor}\label{cor:he:te} 
Let $\mu\ll\varepsilon$ and $j\in\{1,2\}$. Then
\lemit{
	\item 	\label{cor:he}	
			$\onorm{p_j(\na_1\heot)}=\onorm{\big(\na_1\heot\big)p_j}
			\ls \mathfrak{e}(t) N^{-1}\mu^{-\frac12}\varepsilon ,$
			\item[]
			$\onorm{\big(\na_1\heot\big)\cdot\na_jp_j}=\onorm{\ket{\phe(z_j)}\bra{\na\phe(z_j)}(\na_1\heot)}
			\ls\mathfrak{e}(t)N^{-1}\mu^{-\frac12},$	
	\item	\label{cor:te}
			$\onorm{p_j\teot}=\onorm{\teot p_j}\ls \mathfrak{e}(t)\varepsilon^\frac12,$
			\item[] $\onorm{p_j(\na_1\teot)}=\onorm{(\na_1\teot)p_j}\ls \mathfrak{e}(t)\varepsilon^{-\frac12}$,
			\item[]$\onorm{\teot\na_jp_j}=\onorm{\ket{\phe(z_j)}\bra{\na\phe(z_j)}\teot}\ls\mathfrak{e}(t)\varepsilon^{-\frac12}.$
}
\end{cor}

\begin{proof}
This follows immediately from Lemma \ref{lem:pfp}, Lemma \ref{lem:he} and Lemma~\ref{lem:theta}.
\end{proof}

\noindent\emph{Estimate of~\eqref{gamma_b_2:1}.} Define $t_2:=2p_2+q_2(1+\pc_2)$. Then we obtain with $\hat{l}$ from Lemma~\ref{lem:l}
\begin{align*}
|\eqref{gamma_b_2:1}|&
\ls N\left|\llr{\hat{l}t_2\qc_1\psi,\wbot p_1p_2\psi}\right|
=N\left|\llr{\hat{l}t_2\qc_1\psi,\teot\wbot p_1p_2\psi}\right|\\
&=N\int\limits_{\R^{3(N-1)}}\hspace{-0.2cm}\d z^{N-1}\hspace{-0.1cm}\int\limits_{\overline{B_\varepsilon(z_2)}}\hspace{-0.1cm}\d z_1\overline{(\hat{l}t_2\qc_1\psi)(z_1\mydots z_N)}
\te(z_1-z_2)\wb(z_1-z_2)(p_2p_1\psi)(z_1\mydots z_N)
\end{align*}
as $\te(\z_1-\z_2)=1$ for $\z_1-\z_2\in\supp\wb$ and $\supp\te=\overline{B_\varepsilon(0)}$. Thus $\wb(z_1-z_2)=\Delta_1\he(z_1-z_2)$ on the whole domain of integration in the $\d z_1$-integral. Integration by parts in $z_1$ yields
\begin{align}
|\eqref{gamma_b_2:1}|
\ls& N\left|\llr{\hat{l}\qc_1\psi,t_2\teot(\na_1\heot)\cdot\nabla_1p_1p_2\psi}\right|
\label{eqn:gamma:b_2:1:2}\\
&+N\left|\llr{\hat{l}\qc_1\psi,t_2(\na_1\teot)\cdot(\na_1\heot)p_1p_2\psi}\right|
\label{eqn:gamma:b_2:1:3}\\
&+N\left|\llr{\na_1\hat{l}\qc_1\psi,t_2\teot(\na_1\heot)p_1p_2\psi}\right|,\label{eqn:gamma:b_2:1:4}
\end{align}
where the boundary terms vanish because $\te(|\z|)=0$ for $|\z|=\varepsilon$. 
We estimate these expressions by application of Lemma~\ref{lem:GammaLambda:and:O_12}. To this end, we write each term as $\llr{\Gamma,O_{1,2}\Lambda}$, where $\Gamma$ and $\Lambda$ are symmetric in the coordinates $\{2\mydots N\}$. Hence
\begin{align*}
|\eqref{eqn:gamma:b_2:1:2}|
%&=N \left|\llr{\hat{l}\qc_1\psi,\Big(t_2\teot(\na_1\heot)p_2\Big)\cdot\na_1p_1\psi}\right|\\
&\overset{\text{\ref{lem:GammaLambda}}}{\ls}N\norm{\hat{l}\qc_1\psi}
	\Big(\left|\llr{t_2\teot(\na_1\heot)p_2\cdot\na_1p_1\psi,t_3\teoth(\na_1\heoth)p_3\cdot\na_1p_1\psi}\right|\\
	&\hphantom{\leq\norm{\hat{l}\qc_1\psi}\Big(}\quad+N^{-1}\norm{t_2\teot(\na_1\heot)p_2\cdot\na_1p_1\psi}^2 \Big)^\frac{1}{2}\\
&\overset{\text{\ref{lem:O_12}}}{\leq} N\norm{\hat{l}\qc_1\psi}\Big(\norm{p_2\teot(\na_1\heot)t_2\cdot\na_1p_1\psi}^2
		+N^{-1}\norm{t_2\teot(\na_1\heot)p_2\cdot\na_1p_1\psi}^2\Big)^\frac12\\
&\leq N\norm{\hat{l}\qc_1\psi}\Big(\onorm{p_2\teot}^2\onorm{(\na_1\heot)\cdot\na_1p_1}^2
		+N^{-1}\norm{\te}^2_{L^\infty(\R^3)}\onorm{(\na_1\heot)p_2}^2\onorm{\na_1p_1}^2 \Big)^\frac{1}{2}\\
&\ls\mathfrak{e}^3(t)\left(\tfrac{\varepsilon^2}{\mu}\right)^\frac12\left(\varepsilon+N^{-1}\right)^\frac12N^\xi
\end{align*}
by Lemma \ref{lem:a_priori}, Lemma \ref{lem:theta} and Corollary \ref{cor:he:te}. Analogously,
\begin{align*}
|\eqref{eqn:gamma:b_2:1:3}|
%&=N\left|\llr{\hat{l}\qc_1\psi,\Big(t_2(\na_1\teot)\cdot(\na_2\heot)p_2\Big) p_1\psi}\right|\\
&\ls N\norm{\hat{l}\qc_1\psi}\Big(\onorm{p_2(\na_1\heot)}^2\onorm{(\na_1\teot)p_1}^2+N^{-1}\norm{\na\te}^2_{L^\infty(\R^3)}\onorm{(\na_1\heot)p_2}^2\Big)^\frac12\\
&\ls\mathfrak{e}^3(t)\left(\tfrac{\varepsilon^2}{\mu}\right)^{\frac12}\left(\varepsilon+N^{-1}\right)^\frac12N^\xi,\\
|\eqref{eqn:gamma:b_2:1:4}|
%&=N\left|\llr{\na_1\hat{l}\qc_1\psi,\Big(t_2\teot(\na_1\heot)p_2\Big)p_1\psi}\right|\\
&\ls N\norm{\nabla_1\hat{l}\qc_1\psi}\Big(\onorm{p_2\teot}^2\onorm{(\nabla_1\heot)p_1}^2
	+N^{-1}\norm{\te}^2_{L^\infty(\R^3)}\onorm{(\na_1\heot)p_2}^2\Big)^\frac12\\
&\ls\mathfrak{e}^3(t)\left(\tfrac{\varepsilon^2}{\mu}\right)^\frac12\left(\varepsilon+N^{-1}\right)^\frac12N^\xi.
\end{align*}
Hence
\begin{equation*}
|\eqref{gamma_b_2:1}|\ls\mathfrak{e}^3(t)\left(\tfrac{\varepsilon^2}{\mu}\right)^\frac12\left(\varepsilon+N^{-1}\right)^\frac12N^\xi
\ls\mathfrak{e}^3(t)\left(\tfrac{\varepsilon^2}{\mu}\right)^\frac12
\end{equation*}
because  $N^{-\frac12+\xi}<1$ as $\xi<\frac{1}{2}$ and
$\varepsilon^\frac12 N^{\xi}=(\frac{\varepsilon^2}{\mu})^\frac14 N^{\xi-\frac{\beta}{4}}\varepsilon^{\frac{\beta}{2}}\ls 1$
for $\mu\ll\varepsilon$ as $\xi\leq\frac{\beta}{4}$.\\\\

\noindent\noindent\emph{Estimate of~\eqref{gamma_b_2:2}.} Define $t_{12}:=\qp_1\pc_1\qc_2+\qc_1q_2$. Analogously to the estimate of $\eqref{gamma_b_2:1}$,
\begin{align}
|\eqref{gamma_b_2:2}|
\leq& N\left|\llr{\hat{l}t_{12}\psi,\wbot p_1q_2\psi}\right|
=N\left|\llr{\hat{l}t_{12}\psi,\teot\left(\Delta_1\heot\right)p_1q_2\psi}\right|\nonumber\\
\leq &N\left|\llr{\hat{l}t_{12}\psi,\teot(\na_1\heot)\cdot\na_1p_1q_2\psi}\right|\label{eqn:gamma:b_2:2:1}\\
&+N\left|\llr{\hat{l}t_{12}\psi,(\na_1\teot)\cdot(\na_1\heot)p_1q_2\psi}\right|\label{eqn:gamma:b_2:2:2}\\
&+N\left|\llr{\na_1\hat{l}t_{12}\psi,\teot(\na_1\heot)p_1q_2\psi}\right|.\label{eqn:gamma:b_2:2:3}
\end{align}
To estimate \eqref{eqn:gamma:b_2:2:1} to \eqref{eqn:gamma:b_2:2:3}, we apply first Lemma~\ref{lem:commutators:2} to commute $\hat{l}$ next to $q_2$ and use the fact that $\norm{\hat{l}_1q_2\psi}\ls 1$ by Lemma~\ref{lem:l} and Lemma \ref{lem:fqq}. Observing that $t_{12}=t_{12}q_1q_2$ and consequently $\norm{t_{12}\psi}\leq\norm{\qc_1\psi}\leq\varepsilon\mathfrak{e}(t)$ by Lemma~\ref{lem:a_priori:1}, we obtain
\begin{align*}
\eqref{eqn:gamma:b_2:2:1}&=N\left|\llr{t_{12}\psi,q_1q_2\hat{l}\teot(\na_1\heot)\cdot(p_1+q_1)q_2\na_1p_1\psi}\right|\\
&=N\left|\llr{t_{12}\psi,\teot(\na_1\heot)\cdot(\hat{l}_1p_1+\hat{l}q_1)\na_1p_1q_2\psi}\right|\\
&=N\left|\llr{t_{12}\psi,\teot(\na_1\heot)\cdot\na_1p_1\hat{l}_1q_2\psi}\right|\\
&\leq N\norm{t_{12}\psi}\norm{\te}_{L^\infty(\R^3)}\onorm{(\na_1\heot)\cdot\na_1p_1}\norm{\hat{l}_1q_2\psi}\ls\mathfrak{e}^2(t)\left(\tfrac{\varepsilon^2}{\mu}\right)^\frac{1}{2}
\end{align*}
and analogously
\begin{align*}
\eqref{eqn:gamma:b_2:2:2}&=N\left|\llr{t_{12}\psi,(\na_1\teot)\cdot(\na_1\heot)p_1\hat{l}_1q_2\psi}\right|\\
&\leq N\norm{t_{12}\psi}\norm{\na\te}_{L^\infty(\R^3)}\onorm{(\na_1\heot)p_1}\norm{\hat{l}_1q_2\psi}
\ls\mathfrak{e}^2(t)\left(\tfrac{\varepsilon^2}{\mu}\right)^\frac{1}{2},\\
\eqref{eqn:gamma:b_2:2:3}
&=N\left|\llr{\na_1t_{12}\psi,\teot(\na_1\heot)p_1\hat{l}_1q_2\psi}\right|\\
&\leq N\left(\norm{\na_1\qp_1\pc_1\qc_2\psi}+\norm{q_2\na_1\qc_1\psi}\right)\norm{\te}_{L^\infty(\R^3)}\onorm{(\na_1\heot)p_1}
\norm{\hat{l}_1q_2\psi}
\ls\mathfrak{e}^2(t)\left(\tfrac{\varepsilon^2}{\mu}\right)^\frac12
\end{align*}
by Lemma~\ref{lem:theta}, Corollary~\ref{cor:he} and Lemma \ref{lem:a_priori}. \\\\

\noindent\emph{Estimate of~\eqref{gamma_b_2:3}.} Analogously to before,
\begin{align*}
|\eqref{gamma_b_2:3}|\leq& N\left|\llr{\hat{l}\qp_1\qp_2\psi,\pc_1\pc_2\wbot p_1\qc_2\psi}\right|
=N\left|\llr{\hat{l}\qp_1\qp_2\psi,\pc_1\pc_2\teot\left(\Delta_1\heot\right)p_1\qc_2\psi}\right|\\
\leq &N\left|\llr{\hat{l}\qp_1\qp_2\psi,\pc_1\pc_2\teot(\na_1\heot)\cdot\nabla_1p_1\qc_2\psi}\right|\\
&+N\left|\llr{\hat{l}\qp_1\qp_2\psi,\pc_1\pc_2(\na_1\teot)\cdot(\na_1\heot)p_1\qc_2\psi}\right|\\
&+N\left|\llr{\na_1\hat{l}\pc_1\qp_1\qp_2\psi,\pc_2\teot(\na_1\heot)p_1\qc_2\psi}\right|\\
\leq& N\norm{\hat{l}\qp_1\qp_2\psi}\norm{\qc_2\psi}\Big(\norm{\te}_{L^\infty(\R^3)}\onorm{(\na_1\heot)\cdot\nabla_1p_1}+\norm{\na\te}_{L^\infty(\R^3)}\onorm{p_1(\na_1\heot)}\Big)\\
&+N\norm{\na_1\hat{l}\pc_1\qp_1\qp_2\psi}\norm{\te}_{L^\infty(\R^3)}\onorm{p_1(\na_1\heot)}\norm{\qc_2\psi}
\ls\mathfrak{e}^2(t)\left(\tfrac{\varepsilon^2}{\mu}\right)^\frac12
\end{align*}
by Lemma~\ref{lem:a_priori}, Lemma~\ref{lem:theta}, Corollary~\ref{cor:he:te} and Lemma~\ref{lem:fqq}.
\qed

\subsubsection{Proof of the bound for $\gamma_b^{(3)}(t)$}
We estimate~\eqref{gamma_b_3:3} as
$$|\eqref{gamma_b_3:3}|\ls \left|\llr{\hat{l}q_1q_2\psi,|\Phi(x_1)|^2p_1q_2\psi}\right|
\ls \norm{\Phi}^2_{L^\infty(\R)}\norm{\hat{l}q_1q_2\psi}\norm{q_2\psi}
\ls\mathfrak{e}^2(t)\llr{\psi,\hat{n}\psi}$$
by Lemma~\ref{lem:Phi} and Lemma \ref{lem:fqq:3}. 
For~\eqref{gamma_b_3:1} and~\eqref{gamma_b_3:2}, we proceed similarly as in Section \ref{subsec:gamma_b^2} for the quasi one-dimensional interaction $\overline{w}$ instead of the three-dimensional interaction $\wb$. 
\begin{definition}\label{def:1d:eff:int}
Define
\begin{equation}\label{def:wbar}
\overline{w}(x):=\int\limits_{\R^2}\d y_1|\chie(y_1)|^2\int\limits_{\R^2}\d y_2|\chie(y_2)|^2\wb(x,y_1-y_2).
\end{equation}
Further, for $\bo\in[0,1]$, define $\hbo:\R\rightarrow\R$ by
\begin{equation}
\hbo(x):=\begin{cases}
	\int\limits_{-N^{-\beta_1}}^{N^{-\beta_1}}G(x',x)\overline{w}(x')\d x' & \text{for }|x|\leq N^{-\beta_1},\\
	0					& else,
\end{cases}
\end{equation}
where
\begin{equation}
G(x',x):=\tfrac12N^{\beta_1}\begin{cases}
	\left(x'+N^{-\beta_1}\right)\left(x-N^{-\beta_1}\right)  & \text{for } x'<x,\\
	\left(x'-N^{-\beta_1}\right)\left(x+N^{-\beta_1}\right)  & \text{for } x'>x.
\end{cases}
\end{equation}
Besides, define
\begin{equation}
\tb(x):=\begin{cases}
	1 						& \text{for }|x|\leq\mu,\\
	\theta_{\beta_1}(|x|) 	& \text{for }\mu<|x|<N^{-\beta_1},\\
	0						& \text{for }|x|\geq N^{-\beta_1},
\end{cases}
\end{equation}
where $\theta_{\beta_1}:[\mu,N^{-\beta_1}]\rightarrow[0,1]$ is a smooth decreasing function with $\theta_{\beta_1}(\mu)=1$, \hbox{$\theta_{\beta_1}(N^{-\beta_1})=0$} analogously to \eqref{def:theta_smooth}.
As before, we will write
\begin{align*}
\owij&:=\overline{w}(x_i-x_j),\qquad 
\hboij:=\hbo(x_i-x_j),\qquad 
\tbij:=\tb(x_i-x_j).
\end{align*}
\end{definition}

\begin{lem}\label{lem:hbar}
\lemit{
	\item	$\hbo$ solves the boundary-value problem
			\begin{equation}\label{boundary:problem:1d}
			\begin{cases}
				\tfrac{\d^2}{\d x^2}\hbo=\overline{w} 	& \text{ for }\; x\in[-N^{-\beta_1},N^{-\beta_1}],\vspace{0.1cm}\\
				\quad\;\,\hbo=0									&\text{ for }\; |x|=N^{-\beta_1}.
			\end{cases}\end{equation}
	\item	%\norm{\hbo}_{L^\infty(\R)}\ls N^{-1-\beta_1}, &\qquad
			%\norm{\hbo}_{L^2(\R)}\ls N^{-1-\frac{3\beta_1}{2}},\vphantom{\bigg(}\\
			\begin{tabular}{p{4cm}p{5cm}}
			$\norm{\tfrac{\d}{\d x}\hbo}_{L^\infty(\R)}\ls N^{-1},$&$\norm{\tfrac{\d}{\d x}\hbo}_{L^2(\R)}\ls N^{-1-\frac{\beta_1}{2}},$
			\end{tabular}
	
	\item	\begin{tabular}{p{4cm}p{4cm}}
			$\norm{\tb}_{L^\infty(\R)}\leq 1$,&$\norm{\tb}_{L^2(\R)}\ls N^{-\frac{\beta_1}{2}},$			
			\end{tabular}\vspace{0.1cm} \\
			\begin{tabular}{p{4cm}p{4cm}}
			$\norm{\tfrac{\d}{\d x}\tb}_{L^\infty(\R)}\ls N^{\beta_1}$,&$\norm{\tfrac{\d}{\d x}\tb}_{L^2\R)}\ls N^{\frac{\beta_1}{2}}.$
			\end{tabular}
}
\end{lem}
\begin{proof}
Part (a) is evident as $G(x',x)$ is Green's function for the problem \eqref{boundary:problem:1d}. For part (b), we compute for $x\in[-N^{-\beta_1},N^{-\beta_1}]$
\begin{equation*}
\left|\tfrac{\d}{\d x}\hbo(x)\right|=
\tfrac{N^{\beta_1}}{2}\bigg|\int\limits_{-N^{-\beta_1}}^x(x'+N^{-\beta_1})\overline{w}(x')\d x'
+\int\limits_x^{N^{-\beta_1}}(x'-N^{-\beta_1})\overline{w}(x')\d x'\bigg|
\ls\norm{\overline{w}}_{L^1(\R)}\ls N^{-1}
\end{equation*}
since
\begin{align}
\norm{\overline{w}}_{L^1(\R)}&=
\int\limits_\R \d x\int\limits_{\R^2}\d y_1|\chie(y_1)|^2\int\limits_{\R^2}\d y_2|\chie(y_2)|^2\wb(x,y_1-y_2)
\nonumber\\
&\leq \norm{\chie}^2_{L^\infty(\R^2)}\int\limits_{\R^2}\d y_1|\chie(y_1)|^2\norm{\wb}_{L^1(\R^3)}
\ls N^{-1}\label{eqn:norm:wbar}
\end{align}
by \eqref{eqn:int:wb}. The second inequality in (b) follows from this as $\supp\hbo=[-N^{-\beta_1},N^{-\beta_1}]$. Part (c) is shown analogously to Lemma~\ref{lem:theta}.
\end{proof}

%Let us again collect some useful estimates in the following corollary.

\begin{cor}\label{cor:hbar}
Let $j\in\{0,1\}$. Then
\lemit{
	\item	$
			%\onorm{\pp_j\hboot}\ls\mathfrak{e}(t)N^{-1-\frac{3\beta_1}{2}},\\
			\onorm{\pp_j(\tfrac{\d}{\d x_1}\hboot)}\ls\mathfrak{e}(t)N^{-1-\frac{\beta_1}{2}},\qquad
			\onorm{(\tfrac{\d}{\d x_1}\hboot)(\partial_{x_j}\pp_j)}\ls\norm{\Phi(t)}_{H^2(\R)}N^{-1-\frac{\beta_1}{2}},
			$
	\item	$\onorm{\pp_j\left(\tfrac{\d}{\d x_1}\tbot\right)}\ls\mathfrak{e}(t)N^{\frac{\beta_1}{2}}.$
}
\end{cor}
\begin{proof}
This follows immediately from Lemma \ref{lem:pfp:4} and Lemma \ref{lem:hbar}.
\end{proof}

\noindent\emph{Estimate of~\eqref{gamma_b_3:1}.} 
Observing that $\pc_1\pc_2\wbot \pc_1\pc_2=\owot\pc_1\pc_2$, we obtain analogously to the estimate of \eqref{gamma_b_2:1} 
\begin{align}
|\eqref{gamma_b_3:1}|\ls& N\left|\llr{\hat{l}\qp_1\qp_2\psi,\owot p_1p_2\psi}\right|
=N\left|\llr{\hat{l}\qp_1\qp_2\psi,\tbot\left(\tfrac{\d^2}{\d x_1^2}\hboot\right)p_1p_2\psi}\right|\nonumber\\
\leq& N\left|\llr{\hat{l}\qp_1\qp_2\psi,\tbot\left(\tfrac{\d}{\d x_1}\hboot\right)\partial_{x_1}\pp_1\pc_1p_2\psi}\right|
\label{eqn:gamma:b_3:1:2}\\
&+ N\left|\llr{\hat{l}\qp_1\qp_2\psi,\left(\tfrac{\d}{\d x_1}\tbot\right)\left(\tfrac{\d}{\d x_1}\hboot\right)\pp_1\pc_1p_2\psi}\right|
\label{eqn:gamma:b_3:1:3}\\
&+ N\left|\llr{\partial_{x_1}\hat{l}\qp_1\qp_2\psi,\tbot\left(\tfrac{\d}{\d x_1}\hboot\right)\pp_1\pc_1p_2\psi}\right|.\label{eqn:gamma:b_3:1:4}
\end{align}
The boundary terms upon integration by parts vanish as $\tb(\pm N^{-\bo})=0.$ With Lemmata \ref{lem:l:2}, \ref{lem:fqq:3}, \ref{lem:a_priori}, \ref{lem:hbar} and Corollary \ref{cor:hbar}, we conclude
\begin{align*}
\eqref{eqn:gamma:b_3:1:2}&\leq N \norm{\hat{l}\qp_1\qp_2\psi}\norm{\tb}_{L^\infty(\R)}\onorm{\left(\tfrac{\d}{\d x_1}\hboot\right)\pp_2} \onorm{\partial_{x_1}\pp_1}
\ls \mathfrak{e}^2(t)\llr{\psi,\hat{n}\psi}^\frac12 N^{-\frac{\beta_1}{2}},\\
\eqref{eqn:gamma:b_3:1:3}
&\overset{\text{\ref{lem:commutators:2}}}{=}N\left|\llr{\hat{l}^\frac12\qp_1\psi,\Big(\qp_2\left(\tfrac{\d}{\d x_1}\tbot\right)
\left(\tfrac{\d}{\d x_1}\hboot\right)p_2\Big)p_1\hat{l}^\frac12_2 \psi}\right|\\
&\overset{\text{\ref{lem:GammaLambda}}}{\ls}N\norm{\hat{l}^\frac{1}{2}\qp_1\psi}\bigg(
	\left|\llr{\qp_2\left(\tfrac{\d}{\d x_1}\tbot\right)\left(\tfrac{\d}{\d x_1}\hboot\right)p_2p_1\hat{l}^\frac12_2 \psi,
	\qp_3\left(\tfrac{\d}{\d x_1}\tboth\right)\left(\tfrac{\d}{\d x_1}\hbooth\right)p_3p_1\hat{l}^\frac12_2 \psi}\right|\\
	&\hphantom{\overset{\text{\ref{lem:GammaLambda}}}{\ls}\norm{\hat{l}^\frac{1}{2}\qp_1\psi}\Big(}
	+N^{-1}\norm{\qp_2\left(\tfrac{\d}{\d x_1}\tbot\right)\left(\tfrac{\d}{\d x_1}\hboot\right)p_2p_1\hat{l}^\frac12_2 \psi}^2\bigg)^\frac12\\
&\overset{\text{\ref{lem:O_12}}}{\leq} N\norm{\hat{l}^\frac{1}{2}q_1\psi}\Big(
	\onorm{\pp_2\left(\tfrac{\d}{\d x_1}\tbot\right)}^2\onorm{\left(\tfrac{\d}{\d x_1}\hboot\right)\pp_1}^2
	\norm{\hat{l}_2^\frac12 q_2\psi}^2\\
	&\hphantom{\overset{\text{\ref{lem:GammaLambda}}}{\ls}\norm{\hat{l}^\frac{1}{2}\qp_1\psi}\Big(}
	+N^{-1}\norm{\tfrac{\d}{\d x}\tb}^2_{L^\infty(\R)}\onorm{\left(\tfrac{\d}{\d x_1}\hboot\right)\pp_1}^2\onorm{\hat{l}^\frac12_2}^2\Big)^\frac12\\
&\ls \mathfrak{e}^2(t)\llr{\psi,\hat{n}\psi}^\frac12\left(\llr{\psi,\hat{n}\psi}+N^{-1+\beta_1+\xi}\right)^\frac12
\ls\mathfrak{e}^2(t)\left(\llr{\psi,\hat{n}\psi}+N^{-1+\beta_1+\xi}\right),\\
\eqref{eqn:gamma:b_3:1:4}
&\leq N\norm{\partial_{x_1}\hat{l}\qp_1\qp_2\psi}\norm{\tb}_{L^\infty(\R)}\onorm{\left(\tfrac{\d}{\d x_1}\hboot\right)\pp_1}\overset{\text{\ref{cor:fqq:2}}}{\ls}\mathfrak{e}^2(t)N^{-\frac{\beta_1}{2}}.
\end{align*}
Hence
\begin{equation*}
|\eqref{gamma_b_3:1}|\ls\mathfrak{e}^2(t)\left(\llr{\psi,\hat{n}\psi}+N^{-\frac{\beta_1}{2}}+N^{-1+\beta_1+\xi}\right).
\end{equation*}\\

\noindent\emph{Estimate of~\eqref{gamma_b_3:2}.} For this term, we choose $\beta_1=0$. Analogously to the estimate of~\eqref{gamma_b_3:1}, 
\begin{align*}
|\eqref{gamma_b_3:2}|
\ls& N\left|\llr{\hat{l}\qp_1\qp_2\psi,\tzot\left(\tfrac{\d^2}{\d x_1^2}\hzot\right)p_1\pc_2\qp_2\psi}\right|\\
\leq& N\left|\llr{\hat{l}\qp_1\qp_2\psi,\tzot\left(\tfrac{\d}{\d x_1}\hzot\right) \partial_{x_1}\pp_1\pc_1\pc_2\qp_2\psi}\right|\\
&+N\left|\llr{\hat{l}\qp_1\qp_2\psi,\left(\tfrac{\d}{\d x_1}\tzot\right)\left(\tfrac{\d}{\d x_1}\hzot\right)
\pp_1\pc_1\pc_2\qp_2\psi}\right|\\
&+N\left|\llr{\partial_{x_1}\hat{l}\qp_1\qp_2\psi,\tzot\left(\tfrac{\d}{\d x_1}\hzot\right)\pp_1\pc_1\pc_2\qp_2\psi}\right|\\
\leq& N\norm{\hat{l}\qp_1\qp_2\psi}\norm{\qp_2\psi}
	\Big(\norm{\tz}_{L^\infty(\R)}\onorm{\left(\tfrac{\d}{\d x_1}\hzot\right)\partial_{x_1}\pp_1}
	+\norm{\tfrac{\d}{\d x}\tz}_{L^\infty(\R)}\onorm{\left(\tfrac{\d}{\d x_1}\hzot\right)\pp_1}\Big)\\
&+N\norm{\qp_2\psi}\norm{\tz}_{L^\infty(\R)}\onorm{\left(\tfrac{\d}{\d x_1}\hzot\right)\pp_1}\norm{\partial_{x_1}\hat{l}\qp_1\qp_2\psi}\\
\;\overset{\text{\ref{cor:fqq:2}}}{\ls}&\norm{\Phi}_{H^2(\R)}\llr{\psi,\hat{n}\psi}+\mathfrak{e}(t)\llr{\psi,\hat{n}\psi}^\frac12\norm{\partial_{x_1}\qp_1\psi}.
\end{align*}
The estimate $\norm{\partial_{x_1}\qp_1\psi}\ls \mathfrak{e}(t)$ (Lemma~\ref{lem:a_priori:3}) is not sharp enough to see that this expression is small. We need a better control of the kinetic energy, which is established in the following refined energy lemma: 

\begin{lem}\label{lem:E_kin}
Let $\beta\in(0,1)$. Then
$$\norm{\partial_{x_1}\qp_1\psi^{N,\varepsilon}(t)}\ls \exp\left\{\mathfrak{e}^2(t)+\int_0^t\mathfrak{e}^2(s)\d s\right\}
\left(\alpha_\xi(t)+\tfrac{\mu}{\varepsilon}+\left(\tfrac{\varepsilon^2}{\mu}\right)^\frac12+N^{-\beta}+(\tfrac{N}{\varepsilon^2})^{-\eta}\right)^\frac12.$$
\end{lem}

\noindent The proof is given in the next section. As a consequence,
$$|\eqref{gamma_b_3:2}|\ls\mathfrak{e}(t)\exp\left\{\mathfrak{e}^2(t)+\int_0^t\mathfrak{e}^2(s)\d s\right\}
\left(\alpha_\xi(t)+\tfrac{\mu}{\varepsilon}+\left(\tfrac{\varepsilon^2}{\mu}\right)^\frac12+N^{-\beta}+(\tfrac{N}{\varepsilon^2})^{-\eta}\right).$$
\qed

\subsection{Proof of Lemma~\ref{lem:E_kin}.}\label{subsec:E_kin}
We prove a refined bound for the kinetic energy. The basic idea of the proof is comparable to Lemma \ref{lem:a_priori}. However, we estimate the single terms in terms of $\alpha_\xi(t)$ instead of using $\mathfrak{e}^2(t)$. Abbreviating $\psi^{N,\varepsilon}(t)\equiv \psi$ and $\Phi(t)\equiv\Phi$, we obtain
\begin{align}
\alpha_\xi(t)&\geq E^{\psi}(t)-\mathcal{E}^{\Phi}(t)\nonumber\\
=&\norm{\partial_{x_1}\psi}^2-\norm{\Phi'}^2_{L^2(\R)}+\llr{\psi,\left(-\Delta_{y_1}+\tfrac{1}{\varepsilon^2}V^\perp(\tfrac{y_1}{\varepsilon})-	
	\tfrac{E_0}{\varepsilon^2}\right)\psi}\nonumber\\
&+\tfrac{N-1}{2}\llr{\psi,\wbot \psi}-\tfrac{\bb}{2}\llr{\psi,|\Phi(x_1)|^2\psi}\nonumber\\
	&+\tfrac{\bb}{2}\Big(\llr{\psi,|\Phi(x_1)|^2\psi}-\lr{\Phi,|\Phi|^2\Phi}_{L^2(\R)}\Big)
	+\llr{\psi,\Vp(t,\z_1)\psi}-\lr{\Phi,\Vp\big(t,(x,0)\big)\Phi}_{L^2(\R)}\nonumber\\
\geq&\norm{\partial_{x_1}\psi}^2-\norm{\Phi'}^2_{L^2(\R)}
	+\tfrac{1}{2}\llr{\psi,\left((N-1)\wbot -\bb|\Phi(x_1)|^2\right)\psi}\nonumber\\
	&-\tfrac{\bb}{2}\Big|\llr{\psi,|\Phi(x_1)|^2\psi}-\lr{\Phi,|\Phi|^2\Phi}_{L^2(\R)}\Big|-\Big|\llr{\psi,\Vp(t,\z_1)\psi}-\lr{\Phi,\Vp\big(t,(x,0)\big)\Phi}_{L^2(\R)}\Big|\nonumber\\
\gs& \norm{\partial_{x_1}\psi}^2-\norm{\Phi'}^2_{L^2(\R)}
	+\tfrac{1}{2}\llr{\psi,\left((N-1)\wbot -\bb|\Phi(x_1)|^2\right)\psi}- \mathfrak{e}^2(t)\llr{\psi,\hat{n}\psi}-\mathfrak{e}^3(t)\varepsilon\label{eqn:E_kin:1}
\end{align}
as $\llr{\psi,\left(-\Delta_{y_1}+\tfrac{1}{\varepsilon^2}V^\perp(\tfrac{y_1}{\varepsilon})-\tfrac{E_0}{\varepsilon^2}\right)\psi}\geq 0$. The last step follows by Lemma~\ref{lem:psi-Phi}, Lemma \ref{lem:Phi} and Lemma \ref{lem:taylor}, analogously to Section \ref{subsec:gamma_a}.
Further, using that $\norm{\partial_{x_1}\pp_1\psi}^2=\norm{\Phi'}^2_{L^2(\R)}\norm{\pp_1\psi}^2
=\norm{\Phi'}^2_{L^2(\R)}(1-\norm{\qp_1\psi}^2)$, we obtain
\begin{align}
\norm{\partial_{x_1}\psi}^2
&=\norm{\partial_{x_1}\qp_1\psi}^2+\norm{\partial_{x_1}\pp_1\psi}^2+\left(\llr{\partial_{x_1}\qp_1\psi,\partial_{x_1}\pp_1\psi}+c.c.
		\right)\nonumber\\
&\overset{\text{\ref{lem:commutators:4}}}{\geq} \norm{\partial_{x_1}\qp_1\psi}^2+\norm{\Phi'}_{L^2(\R)}^2\left(1-\norm{\qp_1\psi}^2\right)
	-2\left|\llr{\hat{n}^{-\frac12}\qp_1\psi,\partial_{x_1}^2\pp_1(\hat{n}^\frac12\qc_1+\hat{n}^\frac12_1\pc_1)\psi}\right|\nonumber\\
&\overset{\text{\ref{lem:a_priori:2}}}{\gs}\norm{\partial_{x_1}\qp_1\psi}^2+\norm{\Phi'}_{L^2(\R)}^2-\llr{\psi,\hat{n}\psi}
	\left(\mathfrak{e}^2(t)+\norm{\Phi}_{H^2(\R)}\right),\label{eqn:E_kin:2}
\end{align}
where we have used that $\hat{n}_1\ls\hat{n}$ and Lemma \ref{lem:fqq:2}.
\eqref{eqn:E_kin:1} and~\eqref{eqn:E_kin:2} yield
\begin{equation}\label{eqn:E_kin:3}
\norm{\partial_{x_1}\qp_1\psi}^2\ls\norm{\Phi}_{H^2(\R)}\alpha_\xi(t)+\llr{\psi,\left(\bb|\Phi(x_1)|^2-(N-1)\wbot \right)\psi}+\mathfrak{e}^3(t)\varepsilon.
\end{equation}
We estimate the second term of~\eqref{eqn:E_kin:3} by inserting $1=p_1p_2+1-p_1p_2$ into both slots of the scalar product:
\begin{align}
&\llr{\psi,(p_1p_2+1-p_1p_2)\left(\bb|\Phi(x_1)|^2-(N-1)\wbot \right)(p_1p_2+1-p_1p_2)\psi}\nonumber\\
&\qquad=\llr{\psi,p_1p_2\left(\bb|\Phi(x_1)|^2-N\wbot \right)p_1p_2\psi}+\norm{\sqrt{\wbot }p_1p_2\psi}^2
	\label{eqn:E_kin:4}\\
&\quad\qquad+\llr{\psi,(1-p_1p_2)\bb|\Phi(x_1)|^2(1-p_1p_2)\psi}-(N-1)\norm{\sqrt{\wbot }(1-p_1p_2)\psi}^2
	\label{eqn:E_kin:5}\\
&\quad\qquad+\left(\llr{\psi,p_1p_2\bb|\Phi(x_1)|^2(1-p_1p_2)\psi}+c.c.\right)
	\label{eqn:E_kin:6}\\
&\quad\qquad- (N-1)\left(\llr{\psi,p_1p_2\wbot (1-p_1p_2)\psi}+c.c.\right).\label{eqn:E_kin:7}
\end{align}
Making use of $\Gamma(x_1)$ from~\eqref{eqn:Gamma(x_1)}, the first term can be estimated as
\begin{align*}
\eqref{eqn:E_kin:4}
&=\llr{\psi,\pp_1\Gamma(x_1)p_1p_2\psi}+\llr{\psi,p_1p_2(b_{N,\varepsilon}-\bb)|\Phi(x_1)|^2p_1p_2\psi}+ \onorm{\sqrt{\wbot }p_1}^2 \\
&\overset{\text{\ref{lem:pfp:2}}}{\ls} \mathfrak{e}^2(t)\left(\tfrac{\mu}{\varepsilon}+N^{-1}+(\tfrac{N}{\varepsilon^2})^{-\eta}\right)
\end{align*}
by \eqref{eqn:norm:Gamma} and \eqref{eqn:int:wb} with $\eta$ from Definition \ref{def:W}. Note that at this point, it is crucial that $\beta<1$. For the second and third term, note that $1-p_1p_2=q_2+q_1p_2$ and $\norm{\sqrt{\wbot }(1-p_1p_2)}^2\geq 0$. Hence
\begin{align*}
\eqref{eqn:E_kin:5}&\leq \llr{\psi,q_2\bb|\Phi(x_1)|^2q_2\psi}+\llr{\psi,q_1p_2\bb|\Phi(x_1)|^2q_1p_2\psi}
\ls\llr{\psi,\hat{n}\psi}\mathfrak{e}^2(t),\\
\eqref{eqn:E_kin:6}&\leq 2 \left|\llr{\hat{n}^\frac12_1\psi,p_1p_2\bb|\Phi(x_1)|^2p_2q_1\hat{n}^{-\frac12}\psi}\right|
\ls\mathfrak{e}^2(t)\llr{\psi,\hat{n}\psi}
\end{align*}
by Lemma \ref{lem:fqq:1} and Lemma \ref{lem:Phi}. For the last term, observe that $1-p_1p_2=p_1q_2+q_1p_2+q_1q_2$, hence, by symmetry of $\psi$, 
\begin{align}
\eqref{eqn:E_kin:7}\leq&
2N\left|\llr{\psi,p_1q_2 \wbot p_1p_2\psi}\right|+N\left|\llr{\psi,q_1q_2 \wbot p_1p_2\psi}\right|\nonumber\\
\ls& N\left|\llr{\hat{n}^{-\frac12}q_2\psi,p_1\wbot p_1p_2\hat{n}^\frac12_1\psi}\right|\label{eqn:E_kin:8}\\
&+N\left|\llr{\qc_1\psi,q_2(1+\pc_2)\wbot p_1p_2\psi}\right|\label{eqn:E_kin:9}\\
&+ N\left|\llr{\psi,\qp_1\qp_2\pc_1\pc_2\wbot p_1p_2\psi}\right|\label{eqn:E_kin:10}
\end{align}
analogously to the decomposition of~\eqref{eqn:dt_alpha:7}. Using \eqref{eqn:int:wb}, \eqref{eqn:E_kin:8} is easily estimated as
\begin{equation*}
\eqref{eqn:E_kin:8}\overset{\text{\ref{lem:pfp:1}}}{\ls}\mathfrak{e}^2(t)\llr{\psi,\hat{n}\psi}.
\end{equation*}
For~\eqref{eqn:E_kin:9}, we obtain with $t_2:=q_2(1+\pc_2)$, similarly to the estimate of~\eqref{gamma_b_2:1},
\begin{align*}
\eqref{eqn:E_kin:9}
\leq& N\left|\llr{\qc_1\psi,t_2\teot(\na_1\heot)\cdot\nabla_1p_1p_2\psi}\right|\\
&+N\left|\llr{\qc_1\psi,t_2(\na_1\teot)\cdot(\na_1\heot)p_1p_2\psi}\right|
+N\left|\llr{\na_1\qc_1\psi,t_2\teot(\na_1\heot)p_1p_2\psi}\right|\\
\leq& N\norm{\qc_1\psi}\left(\norm{\te}_{L^\infty(\R^3)}\onorm{(\na_1\heot)\cdot\nabla_1p_1}
+\onorm{p_1(\na_1\heot)}\norm{\na\te}_{L^\infty(\R^3)}\right)\\
&+N\norm{\na_1\qc_1\psi}\norm{\te}_{L^\infty(\R^3)}\onorm{p_1(\na_1\heot)}
\ls\mathfrak{e}^2(t)\left(\tfrac{\varepsilon^2}{\mu}\right)^\frac12.
\end{align*}
\eqref{eqn:E_kin:10} is of the same structure as~\eqref{gamma_b_3:1}. Choosing $\bo=\beta$, one computes analogously to~\eqref{eqn:gamma:b_3:1:2} to~\eqref{eqn:gamma:b_3:1:4}
\begin{align*}
\eqref{eqn:E_kin:10}
=&N\left|\llr{\hat{n}^{-\frac12}\qp_1\qp_2\psi,\tbetaot\left(\tfrac{\d^2}{\d x_1^2}\hbot\right)\pp_1\pc_1\hat{n}_2^\frac12p_2\psi}\right|\\
\leq& N\left|\llr{\hat{n}^{-\frac12}\qp_1\qp_2\psi,\tbetaot\left(\tfrac{\d}{\d x_1}\hbot\right)\partial_{x_1}\pp_1\hat{n}_2^\frac12\pc_1p_2\psi}\right|\\
&+N\left|\llr{\partial_{x_1}\hat{n}^{-\frac12}\qp_1\qp_2\psi,\tbetaot\left(\tfrac{\d}{\d x_1}\hbot\right)p_1p_2\hat{n}_2^\frac12\psi}\right|\\
&+N\left|\llr{\hat{n}^{-\frac12}\qp_1\psi,\qp_2\left(\tfrac{\d}{\d x_1}\tbetaot\right)\left(\tfrac{\d}{\d x_1}\hbot\right)p_2\hat{n}_2^\frac12p_1\psi}\right|\\
\overset{\text{\ref{lem:GammaLambda:and:O_12}}}{\ls}&\norm{\Phi}_{H^2(\R)}\llr{\psi,\hat{n}\psi}N^{-\frac{\beta}{2}}
+\mathfrak{e}^2(t)N^{-\frac{\beta}{2}}\llr{\psi,\hat{n}\psi}^\frac12\\
&+N\norm{\hat{n}^\frac12\psi}\left(\norm{p_2(\tfrac{\d}{\d x_1}\tbetaot)(\tfrac{\d}{\d x_1}\hbot)\qp_2\hat{n}^\frac12_2p_1\psi}^2+N^{-1}\norm{\qp_2(\tfrac{\d}{\d x_1}\tbetaot)(\tfrac{\d}{\d x_1}\hbot)p_2\hat{n}^\frac12_2p_1\psi}^2\right)^\frac12\\
\ls&\mathfrak{e}^2(t)\left(\llr{\psi,\hat{n}\psi}+N^{-\beta}\right),
\end{align*}
since $n_2(k)\ls n(k)$ and by Corollary~\ref{cor:fqq:2} and Lemma~\ref{lem:a_priori:3}. Besides, we have used that $N^{-1+\beta}<1$ and $\norm{\Phi}_{H^2(\R)}N^{-\frac{\beta}{2}}\ls\mathfrak{e}^2(t)$ for sufficiently large $N$ at fixed time $t$. Thus,
\begin{equation}
\eqref{eqn:E_kin:7}\ls\mathfrak{e}^2(t)\left(\left(\tfrac{\varepsilon^2}{\mu}\right)^\frac12+N^{-\beta}+\llr{\psi,\hat{n}\psi}\right).
\end{equation}
Finally, inserting the bounds for~\eqref{eqn:E_kin:4} to~\eqref{eqn:E_kin:7} into~\eqref{eqn:E_kin:3} yields
\begin{align*}
\norm{\partial_{x_1}\qp_1\psi}^2&\ls
\norm{\Phi}_{H^2(\R)}\alpha_\xi(t)+\mathfrak{e}^2(t)\left(\left(\tfrac{\varepsilon^2}{\mu}\right)^\frac12+\tfrac{\mu}{\varepsilon}+N^{-\beta}+(\tfrac{N}{\varepsilon^2})^{-\eta}+\llr{\psi,\hat{n}\psi}\right)\\
&\ls\left(\alpha_\xi(t)+\tfrac{\mu}{\varepsilon}+\left(\tfrac{\varepsilon^2}{\mu}\right)^\frac12+N^{-\beta}+(\tfrac{N}{\varepsilon^2})^{-\eta}\right)\exp\left\{2\mathfrak{e}^2(t)+2\int_0^t\mathfrak{e}^2(s)\d s\right\}
\end{align*}
since $\varepsilon<\left(\frac{\varepsilon^2}{\mu}\right)^\frac12$ and $\mathfrak{e}^2(t)\ls \exp\left\{2\mathfrak{e}^2(t)\right\}$.
\qed

%%---BACKMATTER---------------------------------------------------------------------------------------------------------------------------------

\section*{Acknowledgments}
\noindent I thank Stefan Teufel for helpful remarks and for his  involvement in the closely related joint project \cite{GP}.
Helpful discussions with Maximilian Jeblick, Nikolai Leopold, Peter Pickl and Christof Sparber are gratefully acknowledged. This work was supported by the German Research Foundation within the Research Training Group 1838 ``Spectral Theory and Dynamics of Quantum Systems''.

\begin{appendix}
\section{Well-posedness of the effective equation}\label{appendix}
Let $\frac12<r\leq 4 $ and let the initial datum $\Phi_0\in H^r(\R)$. Local existence of $H^{r}$-solutions of \eqref{NLS} on the maximal time interval $t\in[0,T_{r})$ follows from the usual contraction argument on the subset $K:=\{u\in X:\norm{u}_X\leq 2R\}$ of the Banach space $X:=\mathcal{C}\left([0,T]; H^r(\R)\right)$ for some $R>0$ and $T<T_r$, where one uses that the map $f:u\mapsto \bb |u|^2u+\Vp(t,\cdot)u$ is locally Lipschitz continuous on $H^r(\R)$. 
To prove global existence, one shows first that $T_s=T_r$ for all $\frac12<r,s\leq 4$ and concludes from an estimate of $\norm{\Phi(t)}_{H^1(\R)}$ that no blow-up can occur \cite{wahlen_lec}:\\
Let $\frac12<r<s\leq 4$ and $\Phi_0\in H^s(\R)$. Clearly, $T_s\leq T_r$. Assume now $T_s<T_r$. Then $C_{T_s}:=\sup_{t\in[0,T_s]}\norm{\Phi(t)}_{H^r(\R)}<\infty$. Applying twice the inequality 
$$\norm{uv}_{H^s(\R)}\leq C\left(\norm{u}_{H^s(\R)}\norm{v}_{H^r(\R)}+\norm{u}_{H^r(\R)}\norm{v}_{H^s(\R)}\right)$$
and using the fact that $H^s(\R)$ is an algebra, one concludes that for $t\in[0,T_s]$
\begin{align*}
\norm{\Phi(t)}_{H^s(\R)}&\leq\norm{\Phi_0}_{H^s(\R)}+\int\limits_0^t\norm{f(\Phi(s))}_{H^s(\R)}\d s\\
&\leq\norm{\Phi_0}_{H^s(\R)}+C\int\limits_0^t\left(C_{T_s}^2+\norm{\Vp(s,\cdot)}_{H^s(\R)}\right)\norm{\Phi(s)}_{H^s(\R)}\d s.
\end{align*}
Grönwall's inequality implies that $\norm{\Phi(t)}_{H^s(\R)}$ cannot blow up at $t=T_s$, which contradicts $[0,T_s)$ being the maximal time interval where $H^s$-solutions exist. Therefore $T_s=T_r=:T_\mathrm{max}$. Hence for $\Phi_0\in H^2(\R)$, $\Phi(t)\in H^2(\R)$ for $t\in[0,T_\mathrm{max})$. Consequently, \eqref{H^1-bound} implies that $\lim_{t\rightarrow T_\mathrm{max}}\norm{\Phi(t)}_{H^1(\R)}<\infty $, hence $T_1=T_\mathrm{max}=\infty$.
\end{appendix}

\renewcommand{\bibname}{References}
\bibliographystyle{abbrv}
    \bibliography{bib_PhD}
\end{document}